\documentclass[11pt]{article}

\usepackage{amsthm,amsmath,euscript,mathrsfs,multirow,amsfonts,setspace} 
\usepackage{multirow}
\usepackage{natbib}
\usepackage{graphicx}
\usepackage{pict2e}
\usepackage[all]{xy}

\usepackage[top=1in,bottom=0.95in,left=1in,right=1in]{geometry}
\usepackage{mathtools}  
\usepackage{fancyhdr}  
\newcommand{\eref}[1]{(\ref{#1})}

\newcommand{\tref}[1]{Table~\ref{#1}}
\newcommand{\fref}[1]{Figure~\ref{#1}}
\newcommand{\cref}[1]{Chapter~\ref{#1}}

\newcommand{\lemmaref}[1]{Lemma~\ref{#1}}

\newcommand{\thmref}[1]{Theorem~\ref{#1}}

\newtheorem{theorem}{\bf Theorem}

\newtheorem{lemma}{\bf Lemma}

\def\bbE{\mathbb{E}}

\def\bbI{\mathbb{I}}

\def\bbP{\mathbb{P}}

\def\cB{{\mathcal B}}

\def\cG{{\mathcal G}}

\def\cM{{\mathcal M}}
\def\cN{{\mathcal N}}

\def\cS{{\mathcal S}}

\def\bfmath#1{\boldsymbol{#1}}

\def\bfe{{\bfmath{e}}}

\def\bfm{{\bfmath{m}}}
\def\bfn{{\bfmath{n}}}

\def\bfv{{\bfmath{v}}}

\def\bfP{\bfmath{P}}
\def\bfQ{\bfmath{Q}}

\def\bfT{\bfmath{T}}

\def\ds{\ensuremath{\displaystyle}}

\def\EIIn{E_{2\cN}}
\def\EIInn{E_{2\cN\cN}}
\def\EIIs{E_{2\cS}}
\def\EIIb{E_{2\cB}}

\DeclareMathOperator*{\Exp}{Exp}

\usepackage{color, colortbl}
\definecolor{LightGray}{gray}{0.9}

\begin{document}
\renewcommand{\belowcaptionskip}{6mm}
\renewcommand\figurename{{\sc Figure}}
\renewcommand\tablename{TABLE}
\renewcommand*{\thefootnote}{\fnsymbol{footnote}}
\begin{center}
\begin{spacing}{1.5}
{\Large \bf General triallelic frequency spectrum under demographic models with variable population size}
\end{spacing}

\vspace{0.5cm}
{Paul A. Jenkins$^a$,\hspace{5mm} Jonas W. Mueller$^b$, \hspace{5mm} Yun S. Song$^{c, d,}$\footnote{To whom correspondence may be addressed: yss@cs.berkeley.edu}} 

\vspace{0.5cm}
$^a$
Department of Statistics, University of Warwick, Coventry CV4 7AL, UK\\

$^b$
Department of EECS, Massachusetts Institute of Technology, Cambridge, MA 02139, USA\\

$^c$
Department of Statistics, University of California, Berkeley,  CA
94720, USA\\

$^d$
Computer Science Division, University of California, Berkeley,  CA
94720, USA\\
\end{center}

\begin{abstract}    
It is becoming routine to obtain datasets on DNA sequence variation across several thousands of chromosomes, providing unprecedented opportunity to infer the underlying biological and demographic forces. Such data make it vital to study summary statistics which offer enough compression to be tractable, while preserving a great deal of information. One well-studied summary is the site frequency spectrum---the empirical distribution, across segregating sites, of the sample frequency of the derived allele. However, most previous theoretical work has assumed that each site has experienced at most one mutation event in its genealogical history, which becomes less tenable for very large sample sizes. In this work we obtain, in closed-form, the predicted frequency spectrum of a site that has experienced at most \emph{two} mutation events, under very general assumptions about the distribution of branch lengths in the underlying coalescent tree. Among other applications, we obtain the frequency spectrum of a triallelic site in a model of historically varying population size. We demonstrate the utility of our formulas in two settings: First, we show that triallelic sites are more sensitive to the parameters of a population that has experienced historical growth, suggesting that they will have use if they can be incorporated into demographic inference. Second, we investigate a recently proposed alternative mechanism of mutation in which the two derived alleles of a triallelic site are created \emph{simultaneously} within a single individual, and we develop a test to determine whether it is responsible for the excess of triallelic sites in the human genome.
\end{abstract}  

\section{Introduction}
Thanks to the recent advances in DNA sequencing technologies, it has become feasible to obtain data on sequence variation across tens of thousands of chromosomes \citep[e.g.][]{cov:etal:2010, kei:cla:2012, nel:etal:2012, ten:etal:2012}, and hence to study the impact of variants of very low population frequency. Classical models underlying population genetic studies have typically assumed that each site is affected by at most one mutation event in the genealogical history relating a sample, but for very large samples this assumption is less tenable. One must then account for sites experiencing repeat mutations, which skew the site frequency spectrum and can generate triallelic and even quadra-allelic sites. Triallelic sites are therefore becoming increasingly common, appearing as a few percent of segregating sites in large-scale resequencing studies, particularly as the threshold on masking sites below a given minor allele frequency is being reduced. There are now examples of studies that have found an association between a triallelic single nucleotide polymorphism (SNP) and a disease phenotype, including coronary heart disease \citep{cra:etal:2006} and inflammatory bowel disease \citep[discussed in][]{hue:etal:2007}.

Triallelic sites also have potential use in inference using frequency spectrum data. The observed frequency spectrum of \emph{diallelic} sites is well-recognized as an important summary of genomic data, maintaining a great deal of the information encapsulated by the full data while being relatively simple to interpret. It is therefore well-studied: the effects of a host of modeling assumptions on the frequency spectrum have been investigated and many theoretical predictions have been made, typically using either coalescent-based or diffusion-based models. For example, one can obtain analytic results incorporating the effects of a population of varying size \citep{gri:tav:1998, woo:rog:2002, pol:kim:2003, pol:etal:2003}, selection \citep{gri:2003}, and population subdivision with instantaneous migration events \citep{che:2012}. The Poisson random field framework of \citet{saw:har:1992} is attractive in this respect because of its amenability to the incorporation of natural selection \citep{saw:har:1992, bus:etal:2001}. This and other diffusion-based approaches can also be extended to obtain numerical solutions for more complicated underlying population demographic histories, including a single population of variable size \citep{wil:etal:2005, eva:etal:2007, boy:etal:2008} or a hierarchy of splitting subpopulations with restricted migration between them \citep{gut:etal:2009, luk:etal:2011, luk:hey:ip}. Essentially, one writes down the Kolmogorov forward equation for the underlying diffusion approximation and then obtains a numerical solution using finite differences \citep{wil:etal:2005, eva:etal:2007, gut:etal:2009} or spectral methods \citep{luk:etal:2011, luk:hey:ip}. Examples of inference using the frequency spectrum such as these are important because they can help us learn about recent human population history, estimate the strength of natural selection, and calibrate our expectations prior to a disease association study. However, none of these approaches make use of the information from triallelic sites since they rely on an infinite-sites assumption in which triallelic sites are never observed \citep[although see][]{des:plo:2008,son:ste:2012,ste:etal:2013}. There have been some extensions to incorporate recurrent mutations into the theory of the frequency spectrum \citep{sar:2006, hob:wiu:2009,jen:son:2011:TPB, bha:etal:2012}, but with the exception of \citet{sar:2006} these all assume a simple demography of a stationary, panmictic population of constant size.

In this paper, we obtain a closed-form expression for the sample frequency spectrum of a site that has experienced two mutation events, under an extension of the standard coalescent model which allows for very general assumptions about the distribution of times between coalescence events. 
This allows us to obtain predictions for the shape of the frequency spectrum allowing for both recurrent mutations and varying historical population size.

To emphasize the usefulness of our results, we consider two applications. First, we investigate the sensitivity of the triallelic frequency spectrum to the assumed demographic history. In particular, our interest is in the question: How much power to distinguish between different demographic models do we gain looking at a triallelic, rather than diallelic, site? In a manner quantified further below, we show that although triallelic sites are far less abundant than diallelic sites, they have rather greater value \emph{per site} in capturing the effects of demographic history.

This application relies on a frequency spectrum in which the two mutation events arose independently during the genealogical history of the site. Recently however, \citet{hod:eyr:2010} noted that there are approximately twice as many triallelic sites in the human genome as would be expected by chance. They explored a number of potential explanations and ultimately favored the idea of a new mutational mechanism: namely, the simultaneous generation of two new alleles due to mutation within a single individual. Although the precise mechanism is unknown, they suggest the instability of base mismatches as a plausible explanation. For example, a mutation of a G=C base pair to an unstable G=A mismatch could give rise to a further mutation to C=A. DNA replication of this mismatch means the ancestral G=C has given rise to both a derived A=T and a derived C=G base pairing (\fref{fig:hod:eyr:2010}). Another possibility is that both strands of the DNA duplex mutate simultaneously due to a chemical or radiation event. As a second application of our results, we design and implement a frequency spectrum-based test for the hypothesis that a subset of triallelic sites were generated by a \emph{simultaneous} mutation event within a single individual giving rise to the two derived alleles. 
The test allows us to account for variable historical population size explicitly, and when we do so we do not find evidence in favor of the existence of such a mechanism (although, as we discuss below, it is likely that this is further confounded by population subdivision in the samples used).

\begin{figure}[t]
	\centering
\includegraphics[width=0.7\textwidth]{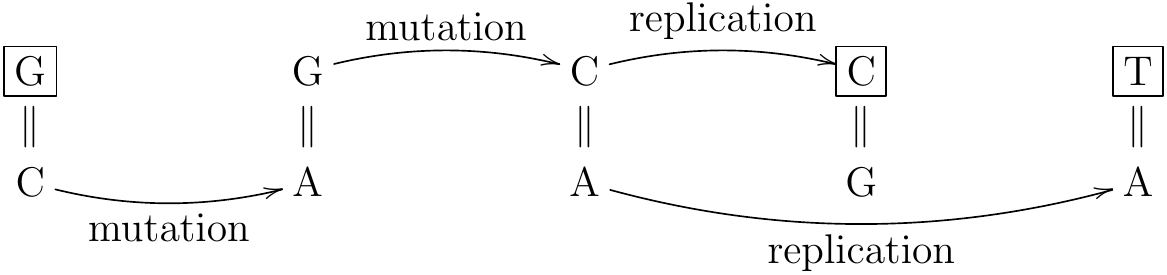}
\caption{\label{fig:hod:eyr:2010}A possible mechanism of simultaneous mutation suggested by \citet{hod:eyr:2010}. Each nucleotide of a single base pair mutates due to instability of the mismatch. Replication of the resulting base pair, using each of the parental strands, results in two new derived alleles. Here, the G allele has given rise to a C and a T (boxed).}
\end{figure}

This paper is organized as follows. In the following section we first introduce some notation and summarize some previous results. We then obtain closed-form formulas for a triallelic frequency spectrum under a general model of coalescence times distributions. In the sections thereafter we consider two applications. First, we perform an extensive simulation study to discern between the sensitivities of diallelic and triallelic frequency spectra to the underlying demographic model. 
Second, we obtain the triallelic frequency spectrum under the proposed simultaneous-mutation mechanism of \citet{hod:eyr:2010}, and develop a likelihood ratio test to compare it with the null triallelic frequency spectrum under independently-occurring mutations.  The test is applied to sequence data taken from the Environmental Genome Project \citep{NIEHS:Aug2011} and the SeattleSNPs project \citep{seattleSNPs:Aug2011} to examine whether some fraction of the triallelic sites in these datasets are in fact the product of simultaneous mutations. 

\section{Notation and previous results}
In this section we introduce our notation and summarize some existing results. Denote by $N_0$ the diploid effective population size in the present generation and by $u$ the probability of a mutation event at a given locus 
per meiosis. For simplicity we assume throughout that the ``locus'' is a single site, although we note that the theory extends easily to other loci that may be of interest. Let $\theta = 4N_0u$ be the population-scaled mutation rate, which we take to be fixed in the usual diffusion limit as $N_0\to\infty$. In this limit and on a timescale of $2N_0$ generations, we denote by $N_t$ the effective population size at time $t$ back in the past, which we take to be a nonrandom function of time such that $N_t \gg 1$ so that a coalescent limit exists for all times \citep[see][for details]{sla:hud:1991, gri:tav:1994:PTRSB}. We assume a general $K$-allele mutation model with mutation transition matrix $\bfP = (P_{ij})$, so that $P_{ij}$ is the probability forwards in time of a mutation taking allele $i$ to allele $j$, given that a type $i$ mutated. It is usual to treat $\bfP$ (and $K$) as fixed and known. We further denote by $a\in\{1,\ldots, K\}$ the ancestral allele at the site of interest, and by $\bfn = (n_1,n_2,\ldots, n_K)$ the unordered sample configuration taken from that site, with total sample size $n = \sum_{i=1}^K n_i$. A unit $K$-vector whose $k$th entry is 1 and all other entries are 0 is denoted by $\bfe_k$. Finally, let $E_s$ denote the event that there were precisely $s$ mutation events at the site in the genealogical history relating the sample.

The sample frequency spectrum can be obtained first by finding the probability of the observed sample configuration under the assumptions of an appropriate coalescent model. This may be partitioned according to the number of mutation events in the genealogical history relating the sample. However, for humans the average per-generation mutation rate for SNPs is small; recent studies show $u \approx 1.2\times 10^{-8}$ \citep{kon:etal:2012,cam:etal:2012}.  Classical population genetics results on the frequency spectrum can be obtained formally by conditioning on precisely one mutation event in the history of the site and then letting $\theta\to 0$. Denoting the ancestral and derived alleles in a diallelic model respectively by $a$ and $b$, it is well known \citep{wat:1975, fu:1995, gri:tav:1998} that for a constant population size, $N_t \equiv N_0$, and a sample configuration of the form $\bfn = (n_a, n_b)$,
\begin{equation}
	\label{eq:SFS}
\phi(i) := \lim_{\theta\to 0}\bbP[(n_a,n_b) = (n-i,i)\mid E_1] = \frac{i^{-1}}{\sum_{j=1}^{n-1}j^{-1}},
\end{equation}
since $\bbP[(n_a,n_b) = (n-i,i),E_1] = \theta i^{-1} + O(\theta^2)$. We refer to the quantity $\phi(i)$ as the \emph{sample} frequency spectrum [as distinguished from the density of the expected number of mutations at each frequency $x\in(0,1)$ in a population of genomes comprising many polymorphic sites, which is also referred to as the (site) frequency spectrum]. Throughout this work we obtain the sample frequency spectrum in a finite-alleles model and in the limit as $\theta \to 0$, after conditioning on the required number of mutation events (for triallelic sites, at least two mutation events are of course necessary). In fact, the result \eref{eq:SFS} is usually obtained by positing a model of \emph{infinitely-many-sites} of mutation, and then finding the distribution of the number of copies of the mutant allele at any random position at which a mutation occurred. Because we condition on looking at a mutant site, this distribution is equivalent to that of a finite-alleles model at a \emph{fixed} site and conditioned on one mutation event, with the implicit assumption that $P_{aa} = 0$ so that the overall rate of mutation in the two models is the same.

There are two extensions to the above result which are relevant to the present work. The first is to general coalescent trees in which the collection of inter-coalescence times, $\bfT = (T_n,T_{n-1},\ldots, T_2)$ is not necessarily given by the standard sequence of independent, exponentially-distributed random variables. In a standard coalescent model we have that the time $T_k$ during which there exist $k$ distinct ancestors to the sample satisfies $T_k \sim \Exp\binom{k}{2}$ on the coalescent timescale. However, certain extensions to this model yield a more complicated distribution for $\bfT$ but leave the topological structure of the tree otherwise unchanged. In this setting, \citet{gri:tav:1998} have obtained the following result:
	Under a coalescent model with general inter-coalescence times $\bfT$ and conditional on precisely one mutation event at a given site, the sample frequency spectrum is given by
	\begin{equation}
	\phi(i) = \frac{\sum_{k=2}^n \alpha_k^{(n-i,i)}\bbE(T_k)}{\sum_{k=2}^n \beta_k\bbE(T_k)},
	\label{eq:gri:tav:1998}
	\end{equation}
where $\alpha_k^{(n-i,i)} = \frac{(n-i-1)!(i-1)!}{(n-1)!}k(k-1)\binom{n-k}{i-1}$ and $\beta_k = k$.
One application of this result is to a coalescent model with a nonconstant population size $N_t$. The distribution for $\bfT$ does not have a simple form, but an expression for $\bbE(T_k)$ is given by \citet{gri:tav:1998}, and an expression for the marginal density of $T_k$ is given by \citet{woo:rog:2002}, \citet{pol:etal:2003} and \citet{pol:kim:2003}. In the Appendix we provide a new proof of \eref{eq:gri:tav:1998}, in order to illustrate our general strategy. For now we merely remark that the topological structure of any polymorphic site having experienced precisely one mutation event in its genealogical history must be of the form shown in \fref{fig:tree}. Coalescent trees of this form are studied in detail by \citet{wiu:don:1999}.

\begin{figure}[t]
	\centering
\includegraphics[width=0.35\textwidth]{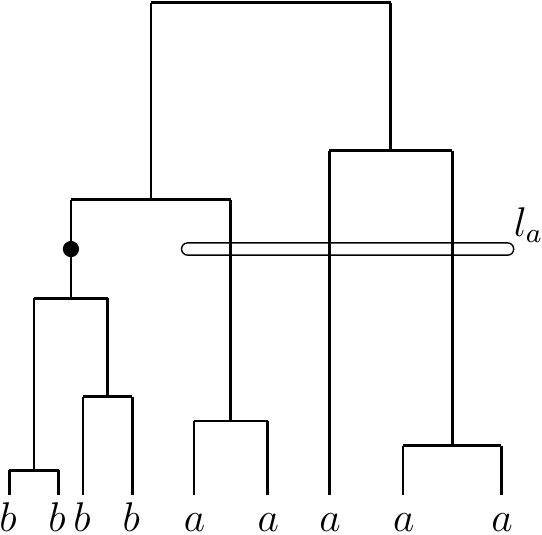}
\caption{\label{fig:tree}A coalescent tree with one mutation. The allele of each leaf is annotated. Also annotated is the variable $l_a$ determining the number of lineages ancestral to allele $a$ at the time of the sole mutation event; in this example, $l_a = 3$.
}
\end{figure}

A second extension of \eref{eq:SFS} is to allow for two mutation events at a polymorphic site. In this case we must consider the exact form of the mutation transition matrix $\bfP$. In particular, it may allow for the second mutation to revert a derived allele to its ancestral state (a \emph{back} mutation) or for the second mutation to create a second independent copy of the extant derived allele (a \emph{parallel} mutation). Such mutations do not give rise to triallelic sites, whereas in practice we will typically identify sites having experienced two mutations only when three alleles are actually observed. Thus, in extending the definition of the sample frequency spectrum to triallelic sites, we condition on \emph{observing} three alleles, an event we denote $O_3$, rather than $E_2$. \citet{jen:son:2011:TPB} have obtained the following result:
		Under a standard coalescent model with $N_t \equiv N_0$ the triallelic sample frequency spectrum is given by
		\begin{align}
		\phi(n_a,n_b,n_c) := {}& \lim_{\theta\to 0}\bbP(\bfn = n_a\bfe_a + n_b\bfe_b + n_c\bfe_c \mid O_3),\notag\\
		= {}& \frac{1}{C}\bigg[P_{ab}P_{bc}d(n_a,n_b,n_c) + P_{ac}P_{cb}d(n_a,n_c,n_b)\notag\\
		& \phantom{\frac{1}{C}\bigg[]}{}+ P_{ab}P_{ac}\left(\frac{1}{n_bn_c} - d(n_a,n_b,n_c) - d(n_a,n_c,n_b)\right)\bigg], \label{eq:jen:son:2011}
		\end{align}
where 
\begin{multline*}
C = \left[\sum_{x\neq a}\sum_{y\neq a,x}P_{ax}P_{xy}\right]\left(H_n + \frac{1}{n} - 2\right)\\ 
	{}+ \left[\sum_{x\neq a}\sum_{y\neq a,x}P_{ax}P_{ay}\right]\left(\frac{(H_{n-1})^2}{2} - \frac{H_{n-1}^{(2)}}{2} - H_n - \frac{1}{n} + 2\right),
\end{multline*}
and
\begin{align*}
	d(n_a,n_b,n_c) &= \frac{1}{(n_a+n_b)(n_a + n_b - 1)}\left[1 + \frac{n}{n_c} - \frac{2n(H_n - H_{n_c-1})}{n_a + n_b + 1}\right],\\
	H_m &= \sum_{j=1}^m \frac{1}{j},\\
	H_m^{(2)} &= \sum_{j=1}^m \frac{1}{j^2}.
\end{align*}
In the above expression $a$, $b$, and $c$ are distinct alleles with $n_a+n_b+n_c = n$; $a$ is the ancestral allele and $b$ and $c$ are derived alleles. We set $H_0 := 0$ by convention.
If the diagonal of $\bfP$ is zero, then the sums involving $\bfP$ in $C$ respectively simplify to $[1-(\bfP^2)_{aa}]$ and $[1-(\bfP\bfP^T)_{aa}]$. [That this simplification requires the diagonal of $\bfP$ to be 0 was inadvertently omitted from \citet[Corollary 6.1]{jen:son:2011:TPB}.]

\section{General triallelic frequency spectrum}

\begin{figure}[t]
\centering
\includegraphics[width=0.65\textwidth]{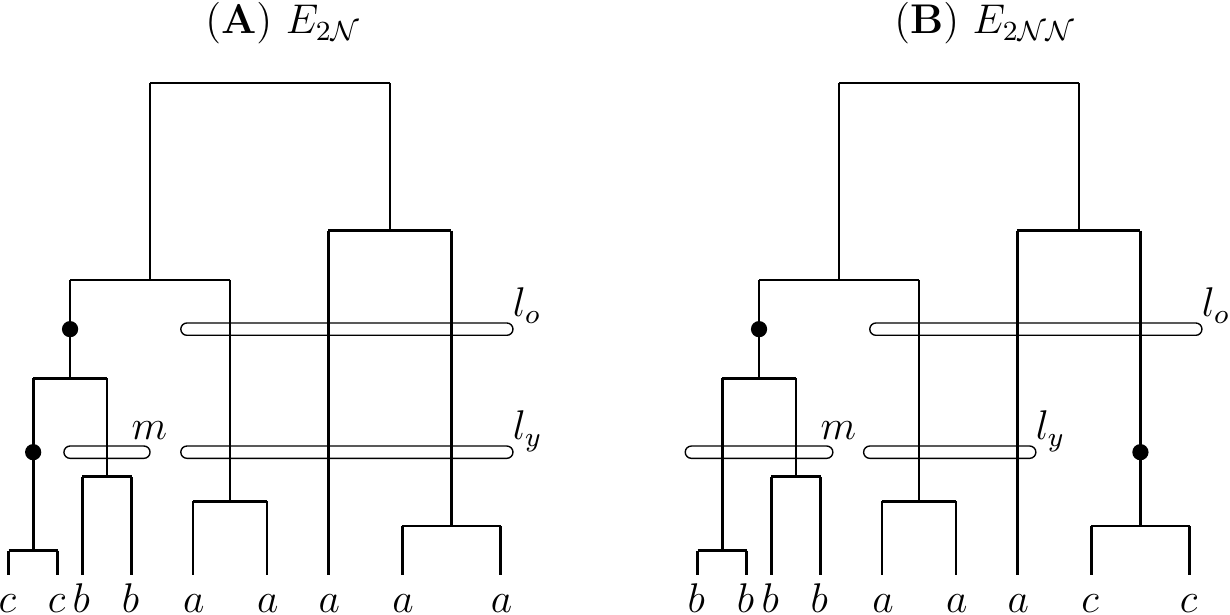}
\caption{\label{fig:trees}Coalescent trees with two mutations. ({\bf A}) Two nested mutations. ({\bf B}) Two nonnested mutations. The allele of each leaf is annotated. Also annotated are variables determining the number of each type at the times of the mutation events; for example, in ({\bf A}) we have $m = 1$, $l_y = 3$, and $l_o = 3$.}
\end{figure}

In this section we obtain a closed-form expression for the sample frequency spectrum of a triallelic site under a general coalescent model with variable population size, $N_t$. This generalizes \eref{eq:gri:tav:1998} to the case of two mutation events at a single site, and generalizes \eref{eq:jen:son:2011} to the case of a variable population size. Our arguments and notation are similar to \citet{jen:son:2011:TPB}, and a brief proof is deferred to an Appendix. In that paper, a key observation is that the event $E_2$ can be partitioned as follows:
\setlength{\leftmargini}{18mm}
\begin{enumerate}
	\item[($\EIIn$)] The two mutation events are genealogically nested.
	\item[($\EIInn$)] The two mutation events are genealogically nonnested and at least one of them does not reside on the basal (adjacent to the root) branches of the tree.
	\item[($\EIIb$)] The two mutation events reside on the two different basal (adjacent to the root) branches of the tree.
	\item[($\EIIs$)] The two mutation events reside on the same branch of the tree.
\end{enumerate}
Only the first two cases can lead to a triallelic site, and so we do not consider the last two any further.  It is straightforward to obtain analogous generalizations for the last two cases, though we omit them. The events $\EIIn$ and $\EIInn$ are illustrated in \fref{fig:trees}. In these examples, the older of the two mutation events gives rise to the allele $b$ and the younger gives rise to the allele $c$, and the subsets of $\EIIn$ and $\EIInn$ satisfying these constraints are denoted by $\EIIn^{(b,c)}$ and $\EIInn^{(b,c)}$ respectively. To find the sample frequency spectrum we first consider the joint probability of observing our triallelic sample with each of these events.

\begin{lemma}
	\label{lem:probs}
Let $n_a\bfe_a + n_b\bfe_b + n_c\bfe_c$ denote a triallelic sample as in \eref{eq:jen:son:2011}. Under a coalescent model with time-dependent population size $N_t$, the joint probability of such a sample together with the way the two mutations are placed on the coalescent tree satisfies
\begin{align}
	\bbP(\bfn = n_a\bfe_a + n_b\bfe_b + n_c\bfe_c,\EIIn^{(b,c)}) &= \frac{\theta^2}{4}P_{ab}P_{bc}\sum_{k=3}^{n_a+n_b+1}\sum_{j=2}^{k-1}C_{j,k}^{(n_a,n_b)} \bbE[T_jT_k] + O(\theta^3),\label{eq:ntwonested}\\
	\bbP(\bfn = n_a\bfe_a + n_b\bfe_b + n_c\bfe_c,\EIInn^{(b,c)}) &= \frac{\theta^2}{4}P_{ab}P_{ac}\sum_{k=3}^{n_a+n_b+1}\sum_{j=2}^k F_{j,k}^{(n_a,n_b)}\bbE[T_jT_k] + O(\theta^3) \label{eq:ntwononnested},
\end{align}
as $\theta\to 0$. Furthermore, 
\begin{align}
	\bbP(\EIIn^{(b,c)}) &= \frac{\theta^2}{4}P_{ab}P_{bc}\sum_{k=3}^n \sum_{j=2}^{k-1}D_{j,k} \bbE[T_jT_k] + O(\theta^3),\label{eq:twonested}\\
	\bbP(\EIInn^{(b,c)}) &= \frac{\theta^2}{4}P_{ab}P_{ac}\sum_{k=3}^n \sum_{j=2}^k G_{j,k}\bbE[T_jT_k] + O(\theta^3). \label{eq:twononnested}
\end{align}                                     
The coefficients in the above expressions are:
\begin{align*}
C_{j,k}^{(n_a,n_b)} = {}& \sum_{l=j-1}^{k-2}\binom{n_a-1}{l-1}\binom{n_b-1}{k-l-2}\binom{k-j}{k-1-l}\binom{n-1}{k-1}^{-1}\binom{k-1}{k-l}^{-1}j(j-1),\\
D_{j,k} = {}& j\left[k\binom{k-2}{j-1} - (j-1)\binom{k-1}{j}\right]\binom{k-1}{j-1}^{-1},\\
F_{j,k}^{(n_a,n_b)} = {}& \sum_{l=(j-2)\vee 1}^{k-2} \binom{n_a-1}{l-1}\binom{n_b-1}{k-l-2}\binom{k-j}{k-2-l}\binom{n-1}{k-1}^{-1}\binom{k-1}{l+1}^{-1}\frac{j(j-1)}{1+\delta_{j,k}},\\
G_{j,k} = {}& \left(k(j-1) - \frac{2\delta_{j,2}}{k-1}\right)\frac{1}{1+\delta_{j,k}},
\end{align*}
where $\delta_{j,k}$ denotes the Kronecker delta.
\end{lemma}
\begin{proof}
	See the Appendix.
\end{proof}
From the above lemma, we can obtain our main result in a straightforward manner.
\begin{theorem}
	\label{thm:main}
	Let $n_a\bfe_a + n_b\bfe_b + n_c\bfe_c$ denote a triallelic sample as above: $a$, $b$, and $c$ are distinct alleles with $n_a+n_b+n_c = n$; $a$ is the ancestral allele and $b$ and $c$ are derived alleles. Under a coalescent model with time-dependent population size $N_t$ and in which mutation events occur independently in the tree, the sample frequency spectrum is
	\begin{align}
		\phi_0(n_a,n_b,n_c) &:= \lim_{\theta\to 0}\bbP(\bfn = n_a\bfe_a + n_b\bfe_b + n_c\bfe_c\mid O_3),\notag\\
		&\phantom{:}= \frac{\ds\sum_{k=3}^{n}\sum_{j=2}^k \gamma_{j,k}^{(n_a,n_b,n_c)}\bbE[T_jT_k]}{\ds\sum_{k=3}^n\sum_{j=2}^k \kappa_{j,k}\bbE[T_jT_k]}, \label{eq:main}
	\end{align}
	where
	\begin{align*}
		\gamma_{j,k}^{(n_a,n_b,n_c)} = {} & (P_{ab}P_{bc}C_{j,k}^{(n_a,n_b)} + P_{ab}P_{ac}F_{j,k}^{(n_a,n_b)})\bbI\{k \leq n_a+n_b+1\}\\
		& {}+ (P_{ac}P_{cb}C_{j,k}^{(n_a,n_c)} + P_{ab}P_{ac}F_{j,k}^{(n_a,n_c)})\bbI\{k \leq n_a+n_c+1\},\\
		\kappa_{j,k} ={} & \left[\sum_{x\neq a}\sum_{y\neq a,x}P_{ax}P_{xy}\right]D_{j,k} + \left[\sum_{x\neq a}\sum_{y\neq a,x}P_{ax}P_{ay}\right]G_{j,k},
	\end{align*}
and $\bbI\{\cdot\}$ denotes the indicator function.
\end{theorem}

\begin{proof}
	As in \citet[Theorem 6.2]{jen:son:2011:TPB}, this follows from
	\begin{align*}
		\bbP(\bfn\mid O_3) &= \frac{\bbP(\bfn,O_3,E_2)}{\bbP(O_3,E_2)} + O(\theta),\\
		&= \frac{\bbP(\bfn,\EIIn^{(b,c)}) + \bbP(\bfn,\EIIn^{(c,b)}) + \bbP(\bfn,\EIInn^{(b,c)}) + \bbP(\bfn,\EIInn^{(c,b)})}{\ds\left[\sum_{x\neq a}\sum_{y\neq a,x}\bbP(\EIIn^{(x,y)})\right] + \left[\sum_{x\neq a}\sum_{y\neq a,x}\bbP(\EIInn^{(x,y)})\right]} + O(\theta).
	\end{align*}
	Now substitute for each term on the right-hand side using \lemmaref{lem:probs} and let $\theta\to 0$.
\end{proof}
Thus, while the frequency spectrum for a site experiencing one mutation event depends only on the first moments of the elements of $\bfT$ (see \eref{eq:gri:tav:1998}), the frequency spectrum for a site experiencing two mutation events depends only on the second moments of the elements of $\bfT$ (\thmref{thm:main}). These moments are considered in further detail by \citet{pol:etal:2003} and \citet{ziv:wie:2008}. Under a suitable choice of historical population size function, $N_t$, the frequency spectrum given by equation \eref{eq:main} will serve as our null model for triallelic sites. As a check on equation \eref{eq:main}, we can fix the population size, $N_t \equiv N_0$, so that
\[
\bbE[T_jT_k] = (1+\delta_{j,k})\binom{j}{2}^{-1}\binom{k}{2}^{-1}.
\]
Inserting this expression into \eref{eq:main} leads to \eref{eq:jen:son:2011}, after extensive simplification.

While this article was under review we learned of related work by \citet{sar:2006}, who also obtains an expression for the frequency spectrum of a site experiencing two mutation events under an arbitrary distribution on $\bfT$ \citep[Lemma 34]{sar:2006}. Our work strengthens his result, which relies on higher order and exponential moments of the elements of $\bfT$. Our work also allows for a more general model of mutation and disentangles the relative contributions of nested and nonnested mutations.

\section{Application I: Sensitivity to demography}

In this section, we compare the frequency spectrum of a diallelic site with that of a triallelic site. 
Given a sample taken from a population whose recent history is described by a demographic model $\cM_1$, 
we can measure the information that is lost if one erroneously applies the frequency spectrum according to another model $\cM_0$. 
To quantify this difference in information, we employ the Kullback-Leibler (KL) divergence, a measure defined for this task \citep{kul:lei:1951, bur:and:2002}. We define the KL divergence from $\cM_1$ to $\cM_0$ by
\[
D(\cM_1 || \cM_0) = \bbE_{\cM_1}\left(\log \frac{\phi_{\cM_1}(\bfn)}{\phi_{\cM_0}(\bfn)}\right),
\]
where $\phi_{\cM_i}$ is the appropriate sampling distribution under model $\cM_i$ and $\bbE_{\cM}$ denotes expectation with respect to random samples $\bfn$ drawn under model $\cM$. Thus, KL divergence is the expected likelihood ratio when testing an alternative $\cM_1$ against a null $\cM_0$ and the alternative is true. Although KL divergence properly refers to distributions under these models rather than the models themselves, when we refer to the divergence between two models it should be clear from the context that we are referring to either their diallelic or triallelic sample frequency spectrum.

We will focus on the divergence from one model of population growth to another. The KL divergence (amount of information loss) 
can thus be compared for samples from a diallelic site versus samples from a triallelic site. 
A larger value of KL divergence for triallelic sites would suggest that such sites are potentially very informative for demographic inference. Throughout this section, a symmetric mutation matrix is used in the frequency spectra calculations (i.e., it is assumed that all transitions between alleles are equally likely). To illustrate how divergences vary at different scales, we focus our analysis on frequency spectra for samples of $10$ and $100$ individuals [a typical magnitude of sample sizes in demographic inference studies \citep{wil:etal:2005, gut:etal:2009}].  

\subsection{The effect of sample size in simulations} In order to compute frequency spectra under general models of historical population size we need first- and second-order moments of $\bfT$ (c.f., Equation \eref{eq:gri:tav:1998} and \thmref{thm:main}). We pre-compute these by simulating coalescent trees using \texttt{ms} \citep{hud:2002}. In order to investigate the effect of sample size in simulations and of the differing dimensions of the two frequency spectra, we first consider the case of a constant population size (i.e., $N_t \equiv N_0$), for which the expected frequency spectra are known in closed form (see Equations~\eref{eq:SFS} and~\eref{eq:jen:son:2011}). Specifically, we compute the KL divergence from the expected frequency spectrum computed exactly to the expected frequency spectrum obtained by simulation of $N_\text{trees}$ to approximate the first- and second-order moments of $\bfT$.  The results are shown in \fref{fig:convergence}.  Although it might be considered unfair to compare KL divergences in diallelic frequency spectra (one-dimensional distributions) with those between triallelic spectra (two-dimensional distributions), \fref{fig:convergence} clearly shows that these divergences exhibit extremely similar behavior as we increase $N_\text{trees}$.  The degree to which the true  spectrum is approximated by its Monte Carlo counterpart is almost exactly the same in the di- and tri-allelic cases, regardless of the choice of $N_\text{trees}$. We also note that for both di- and tri-allelic spectra, the KL divergences have similar magnitude between sample sizes $10$ and $100$ for all choices of $N_\text{trees}$. Thus, KL divergence appears to be a good measure of the difference between two frequency spectra which is relatively invariant to the differences in dimensionality and sample size in our study. Because \fref{fig:convergence} illustrates that using $N_\text{trees} = 10^6$ results in negligibly small divergences (on the order of $10^{-8}$) from the true (closed-form) frequency spectra, we fix $N_\text{trees} = 10^6$ in the remainder of this section.

\begin{figure}[t]
\centering
\includegraphics[width = 0.4\textwidth]{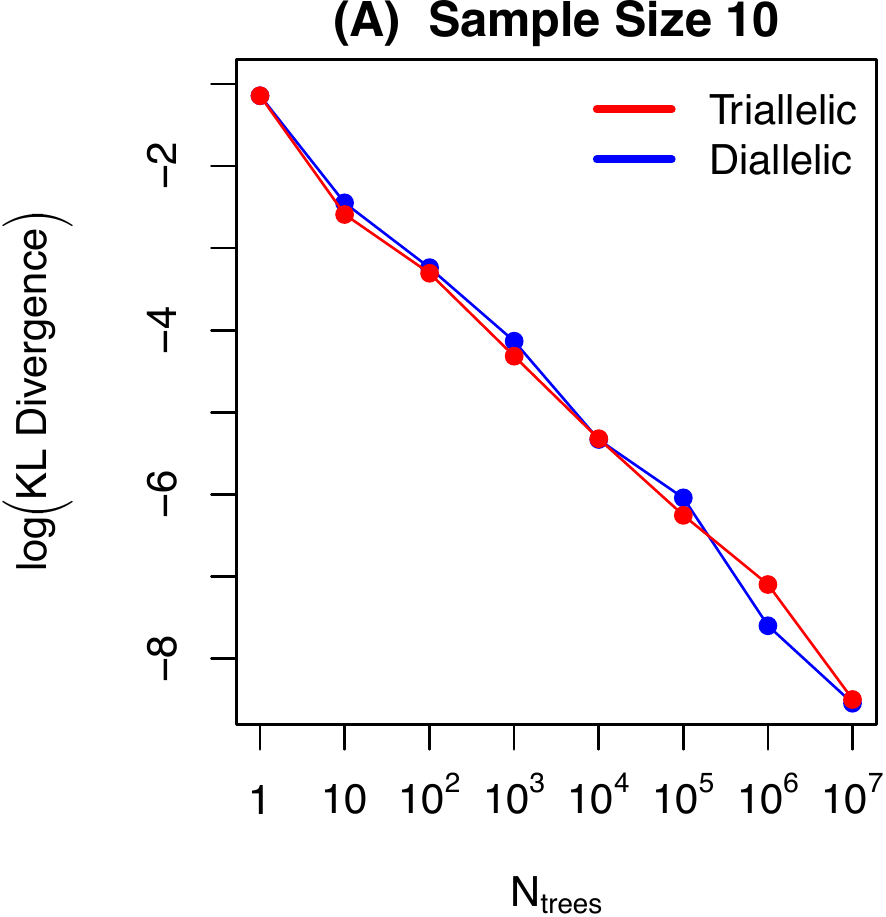} \hspace{1cm}
\includegraphics[width = 0.4\textwidth]{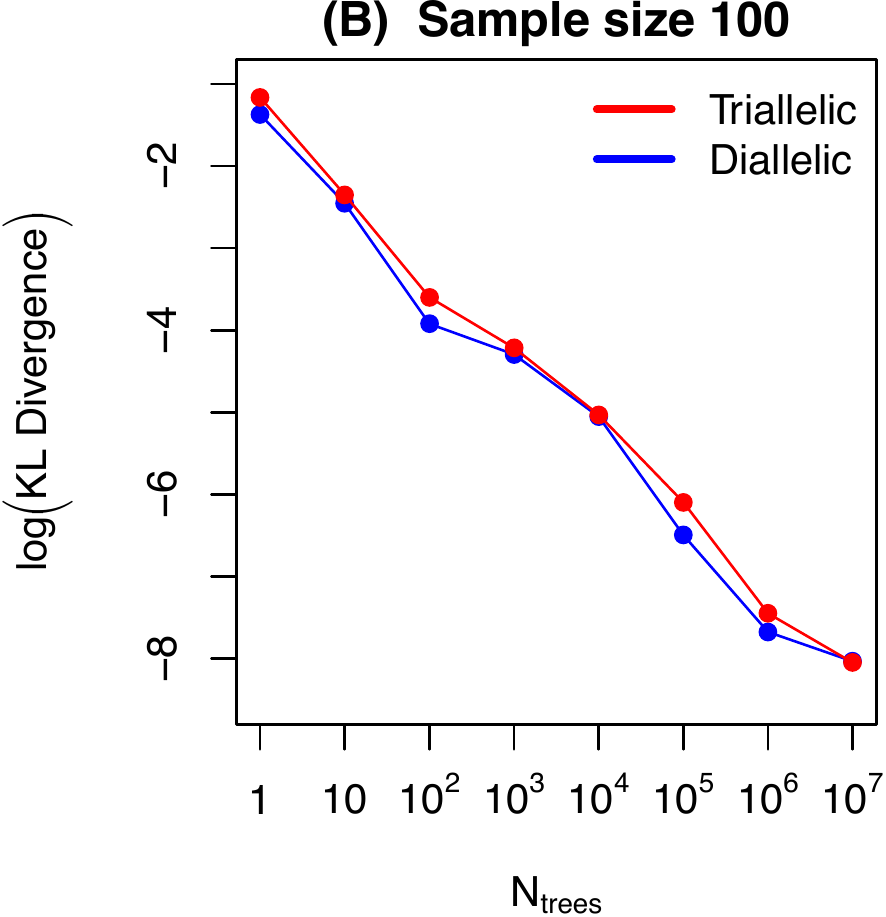} 
\caption{\label{fig:convergence}
Statistical error in the frequency spectrum computation due to approximating first- and second-order moments of inter-coalescence times using simulations. 
For a constant population size, the plots show mean (base $10$ log) KL divergence (over $10$ repetitions) of the frequency spectrum approximated using simulations of inter-coalescence times, from the true frequency spectrum computed by using the exact expected inter-coalescence times. 
The diallelic frequency spectrum is shown in blue and the triallelic frequency spectrum is shown in red, for samples of size (A) $10$ and (B) $100$.
}
\end{figure}

\subsection{Exponential growth} Next, we examine how sampling from a population with historical exponential growth affects the resulting frequency spectrum. As specified in \fref{fig:expgrowthmodels}, we investigate seven models of exponential growth, $\cG_i$, $i=1,\ldots 7$ (with $\cG_0$ representing a population of fixed size). To compute their respective moments of inter-coalescence times, we first simulate $10^6$ trees from populations according to each model. In \fref{fig:expgrowth}A and \ref{fig:expgrowth}B we compute the KL divergence $D(\cG_i || \cG_0)$ of the sample frequency spectrum under $\cG_0$ from the sample frequency spectrum under $\cG_i$, for $i = 1,\ldots 7$. To investigate the potential benefit of triallelic spectra in fine-tuning between two population growth models with different degrees of exponential growth, we also compute KL divergences $D(\cG_i || \cG_{i-1})$ of sample frequency spectra under growth model $\cG_{i-1}$ from growth model $\cG_i$, for $i = 1,\dots,7$ (\fref{fig:expgrowth}C, \ref{fig:expgrowth}D).

\begin{figure}[t] 
\centering
\includegraphics[width = 0.5 \textwidth]{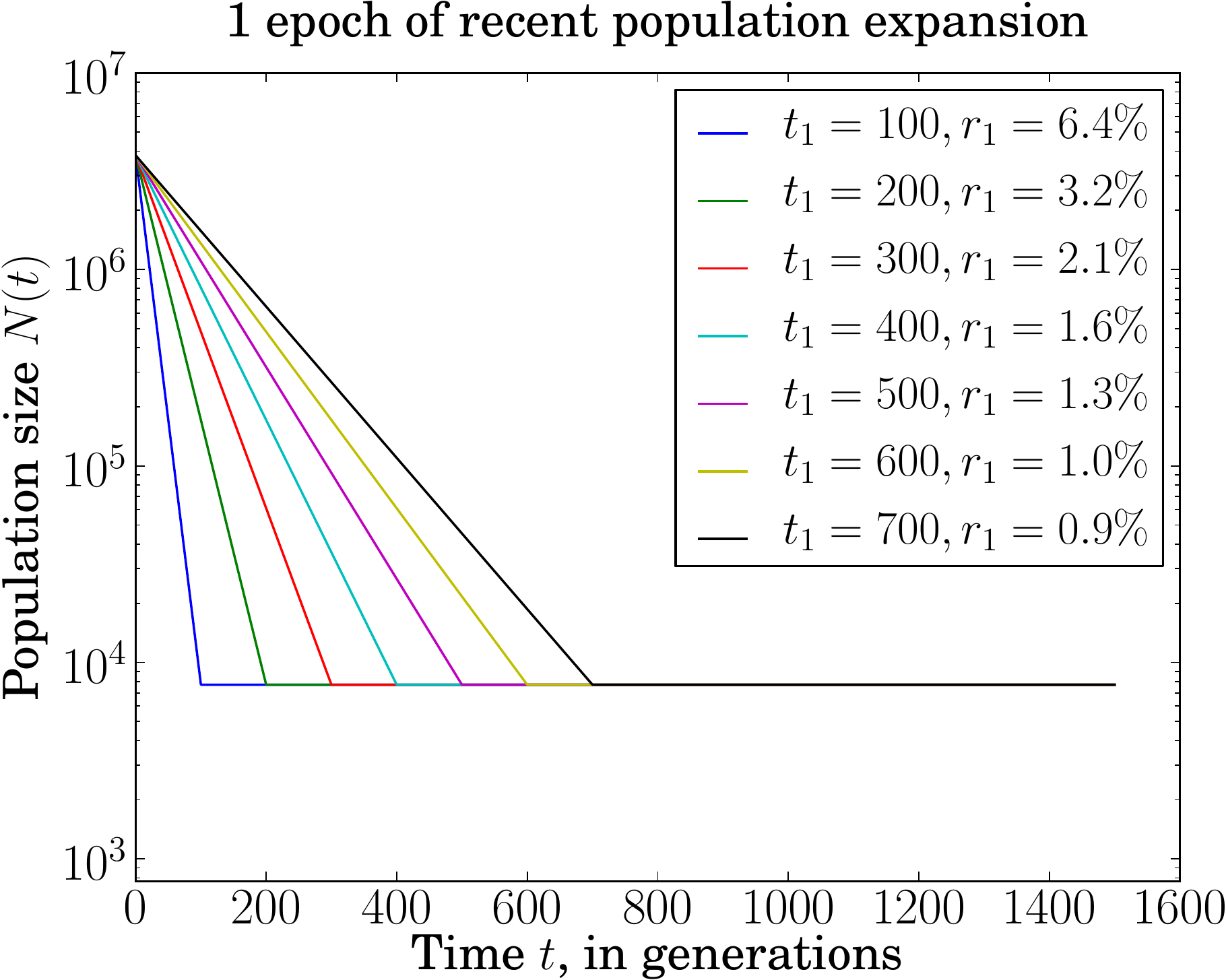}
\caption{\label{fig:expgrowthmodels}Seven models of exponential population growth examined in our analysis. Models 1--7 correspond to growth curves with the onset of growth occurring from $t_1 = 100$ generations ago to $t_1 = 700$ generations ago, respectively. The growth rate per generation, $r_1$,  is chosen so that $N_0 = 10^{6.5}$ in each model. 
}
\end{figure}

\fref{fig:expgrowth} demonstrates clear superiority of using triallelic spectra to distinguish between demographic models with varying degrees of exponential population growth. The mean KL divergence from exponential growth model $i$ to $i-1$ is increased by 87\% when triallelic spectra are used in place of diallelic spectra for samples of size $10$ (and the divergence is increased by 99\%  for sample size $100$). This indicates that triallelic sites contain information which may significantly increase our ability to discern between competing exponential growth models with similar parameters.

\begin{figure}[t]
\centering
\includegraphics[width = 0.48 \textwidth]{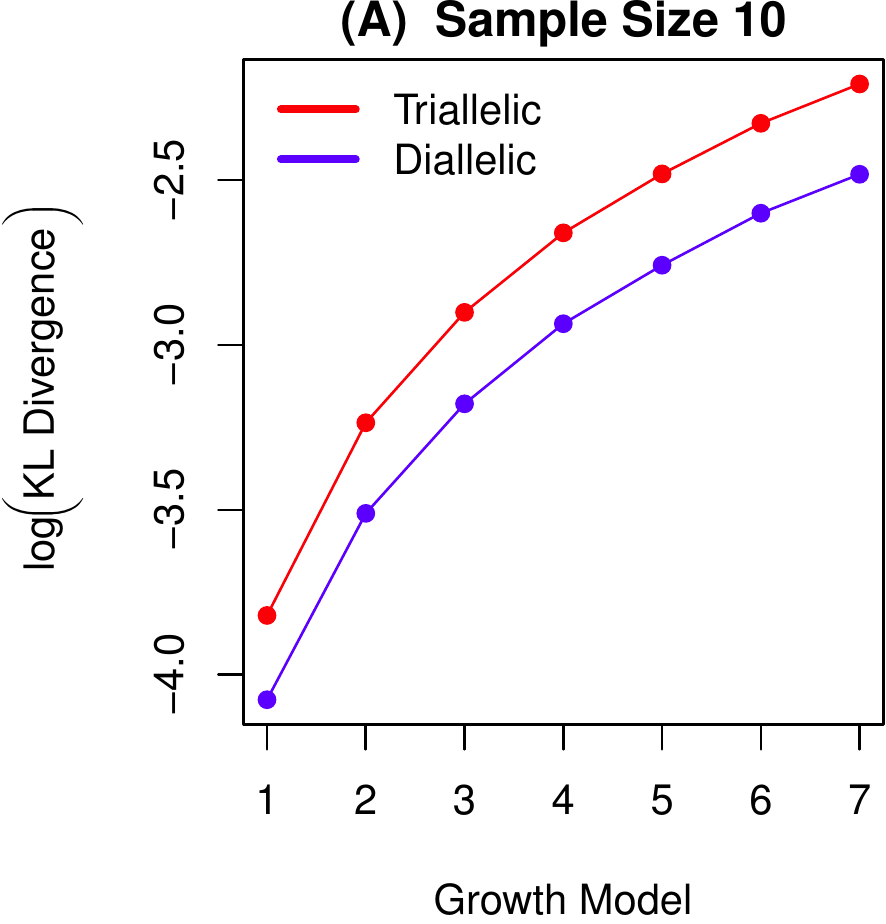}\hfill
\includegraphics[width = 0.48 \textwidth]{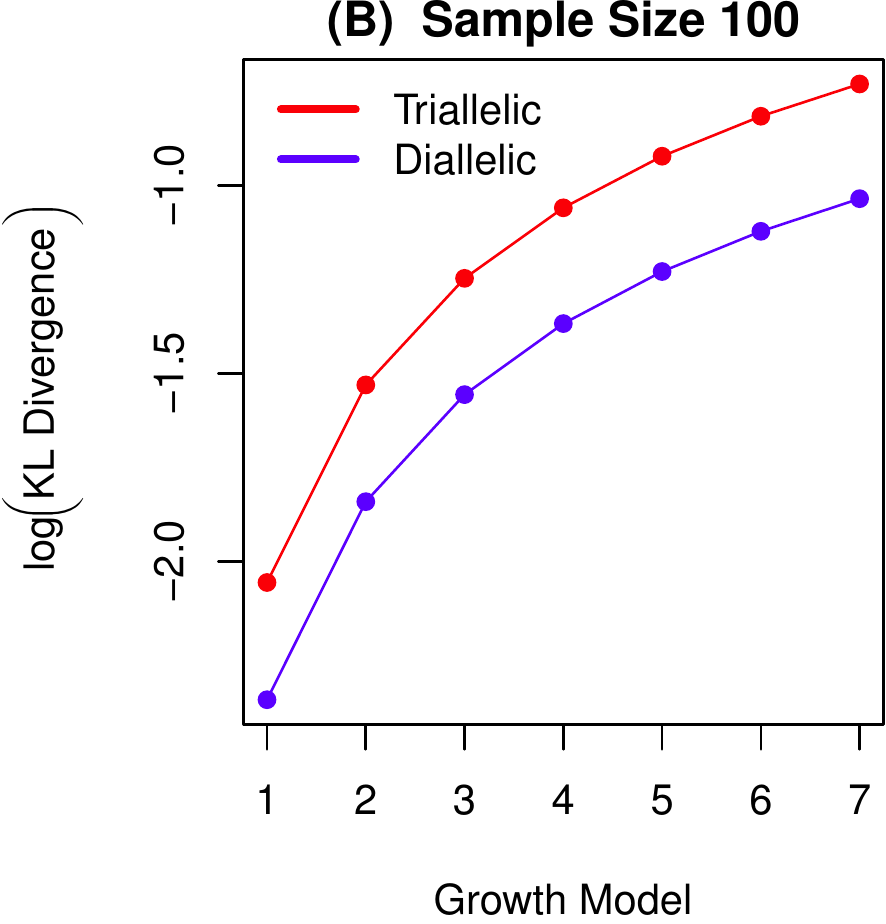}\\
~\\
\includegraphics[width = 0.48 \textwidth]{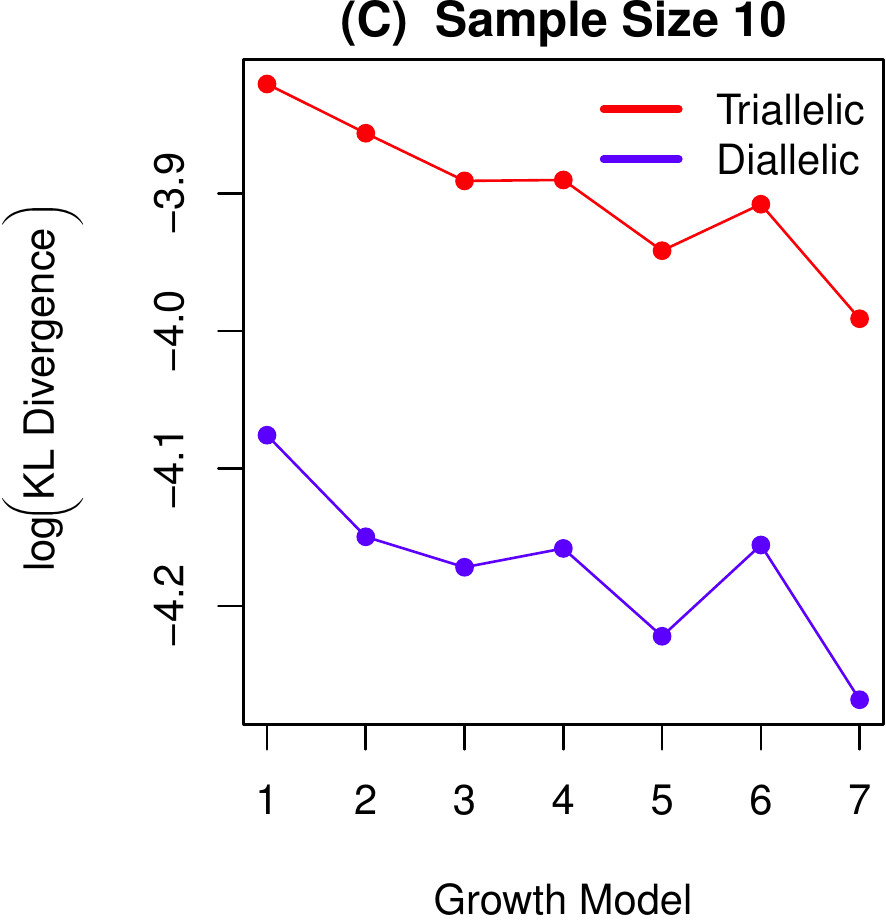}\hfill
\includegraphics[width = 0.48 \textwidth]{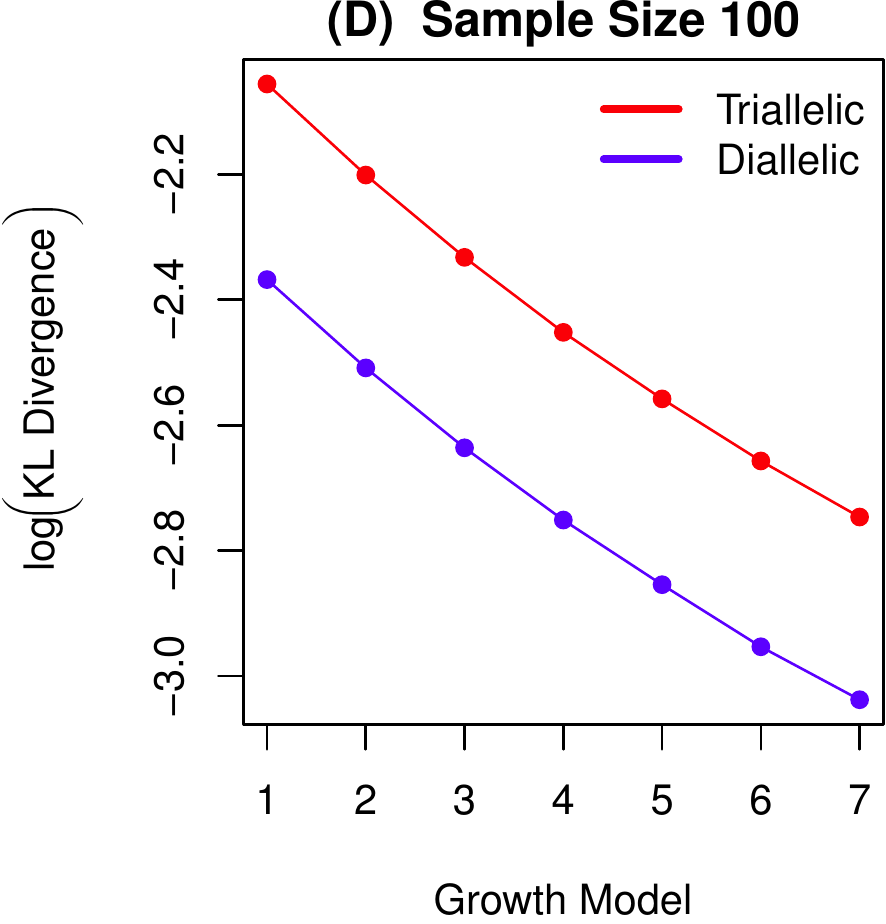}
\caption{\label{fig:expgrowth}KL divergence from one model of exponential population growth to another. (A) and (B) depict the (base $10$ log) KL divergence $D(\cG_i || \cG_0)$ of the sample frequency spectrum taken from a population under growth model $\cG_i$ from the sample frequency spectrum taken from a population of fixed size ($\cG_0$).  (C) and (D) depict the (base $10$ log) KL divergence $D(\cG_i || \cG_{i-1})$ of the sample frequency spectrum taken from a population under growth model $\cG_i$ from the sample frequency spectrum taken from a population under growth model $\cG_{i-1}$. Growth models $\cG_i$, $i = 1,\ldots, 7$ are defined in \fref{fig:expgrowthmodels}; results are shown for the diallelic (blue) and triallelic (red) sample frequency spectrum. 
}
\end{figure}

\subsection{Instantaneous growth} We next investigate in further detail the effect of sample size on KL divergence from a growth model to a model of fixed population size. For the growth model we assumed a function of historical human population growth as inferred by \citet{wil:etal:2005}, who assumed an instantaneous expansion of the population from an ancestral size $\nu N_0$ to a modern size $N_0$ a time $\tau$ ago. Using data from the Environmental Genome Project and working within the framework of the Poisson random field model \citep{saw:har:1992}, \citet{wil:etal:2005} inferred maximum likelihood estimates (MLEs) of $\hat{\nu} = 0.160$ and $\hat{\tau} = 0.00885$, the latter in units of $2N_0$ generations. (The authors estimated $N_0 \approx 51,340$ directly by comparing polymorphism and divergence data, in which case the latter estimate is calibrated as $\hat{\tau} = 908$ generations, or, further assuming a 20 year generation time, $\hat{\tau} \approx 18,200$ years.) We denote this model by $\cG_W$. Thus, given samples that actually stem from a population such as the one described by \citet{wil:etal:2005}, the KL divergence $D(\cG_W || \cG_0)$ quantifies the ability to distinguish that these samples do not come from a fixed-size population. This analysis therefore studies the effect of sample size on the ability to perform inference of population growth parameters under realistic settings in a problem of great interest, for both di- and tri-allelic sites.

\fref{fig:sampsize} illustrates that the KL divergence using triallelic spectra for the two models is much greater, for any sample size, than the divergence using the corresponding diallelic spectra. Furthermore, we see that the increase in KL divergence which results from the presence of the third allele at a triallelic site also grows with increasing sample size, at least up to sample sizes of $25$, before leveling off beyond $25$ and providing a consistent $\sim 97\%$ increase in KL divergence.

\begin{figure} \centering
\includegraphics[width =  0.75\textwidth]{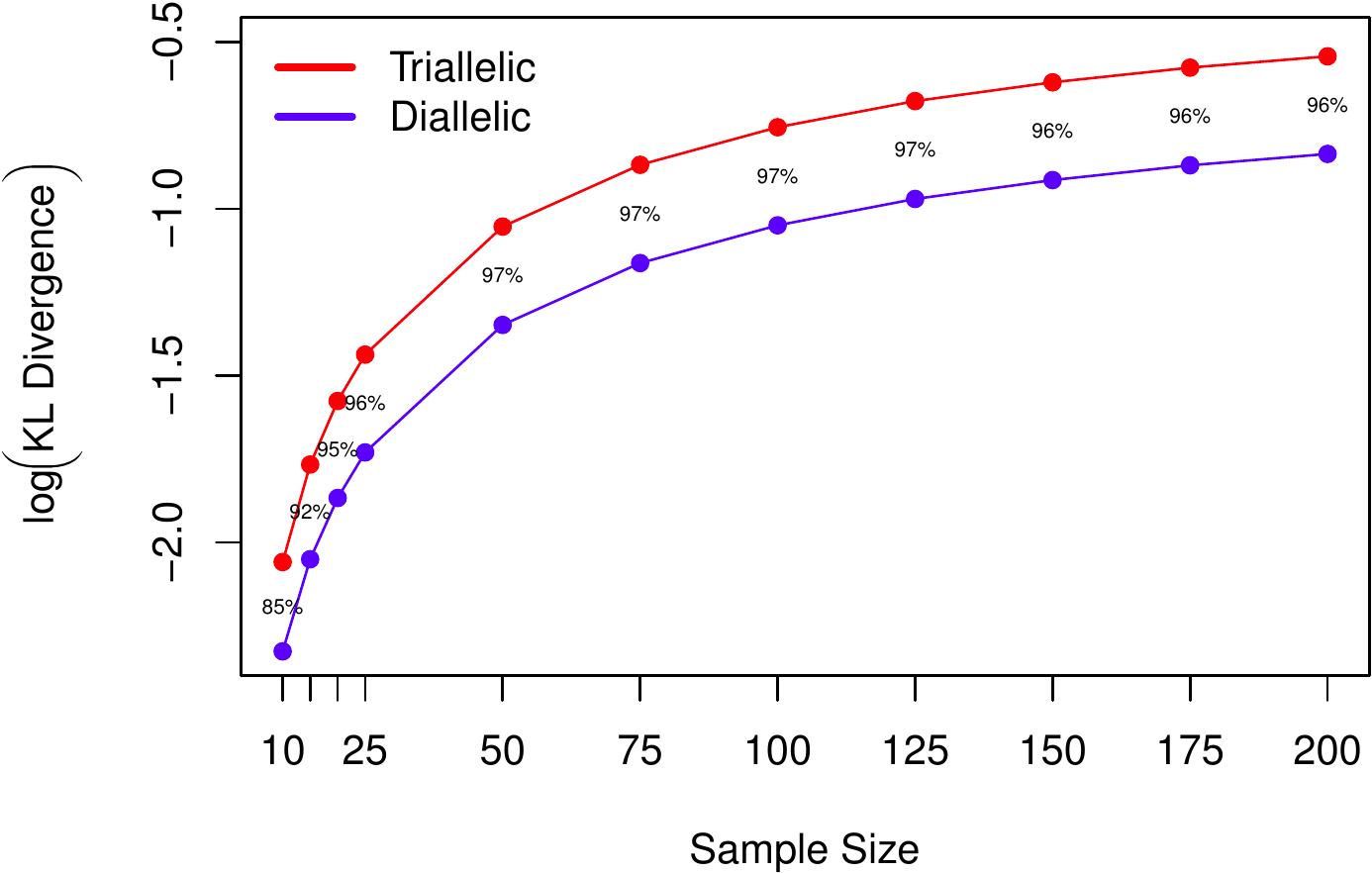}
\caption{\label{fig:sampsize}KL divergence $D(\cG_W || \cG_0)$ from a population model with the growth function of \citet{wil:etal:2005} to a population model of fixed size, for various sample sizes and using di- (blue) and tri- (red) allelic frequency spectra. The percentages between the curves denote the percent increase in KL divergence that results from computing divergence using triallelic spectra rather than diallelic spectra. 
}
\end{figure}

We further address the effect of the parameters of a model of instantaneous population size change, as follows. Adopting $\cG_W$ as a reference point, we examine variations in the two parameters $\nu$ and $\tau$. First the amount of instantaneous growth is varied while keeping the time of the size-change fixed to same value as $\cG_W$ (see Models 8--20 in \tref{tab:instgrowthmodels}), and subsequently, different times of size-change occurrence are examined while the amount of instantaneous growth is fixed so that the pre-change population size is $15\%$ of the post-change size (Models 21--29 in \tref{tab:instgrowthmodels}).

\begin{table}[t] 
\caption{\label{tab:instgrowthmodels}22 models of instantaneous population size-increase from $\nu N_0$ to $N_0$ at a fixed historical time point $\tau$. In models 8--20 the time $\tau$ is fixed at the same value as in the human population growth model proposed by \citet{wil:etal:2005} ($\tau = 0.0044$ in units of $4N_0$ generations), while in models 21--29 the magnitude of growth is fixed at $\nu = 0.15$.}
\begin{center}
\begin{tabular}{r|ccccccccccccc}
\hline
Model & 8 & 9 & 10 & 11 & 12 & 13 & 14 & 15 & 16 & 17 & 18 & 19 & 20\\
\hline
$\nu$ & 0.9 & 0.8 & 0.7 & 0.6 & 0.5 & 0.4 & 0.3 & 0.25 & 0.2 & 0.15 & 0.1 & 0.05 & 0.01\\
\hline 
\end{tabular}
\\[20pt]

\begin{tabular}{r|ccccccccc}
\hline
Model & 21 & 22 & 23 & 24 & 25 & 26 & 27 & 28 & 29\\
\hline
$\tau$ & 1.0 & 0.5 & 0.1 & 0.05 & 0.01& 0.005 & 0.001 & 0.0005 & 0.0001\\
\hline 
\end{tabular}
\end{center}
\end{table}

\begin{figure}[p]
\includegraphics[width = 0.48 \textwidth]{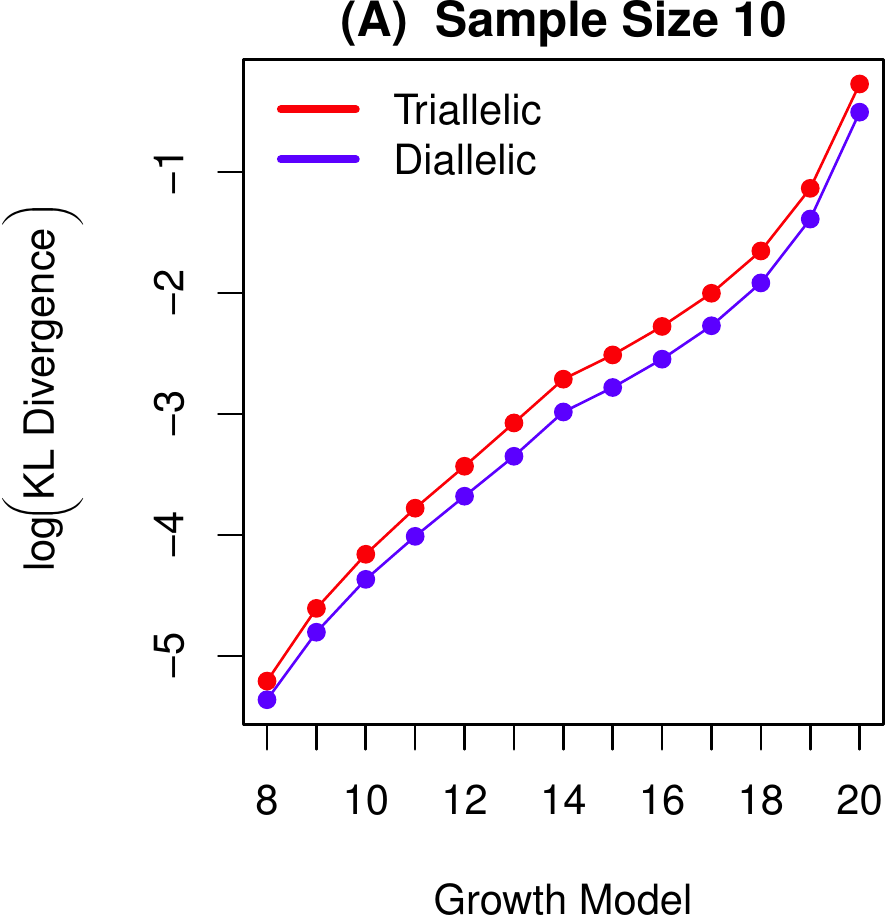}\hfill
\includegraphics[width = 0.48 \textwidth]{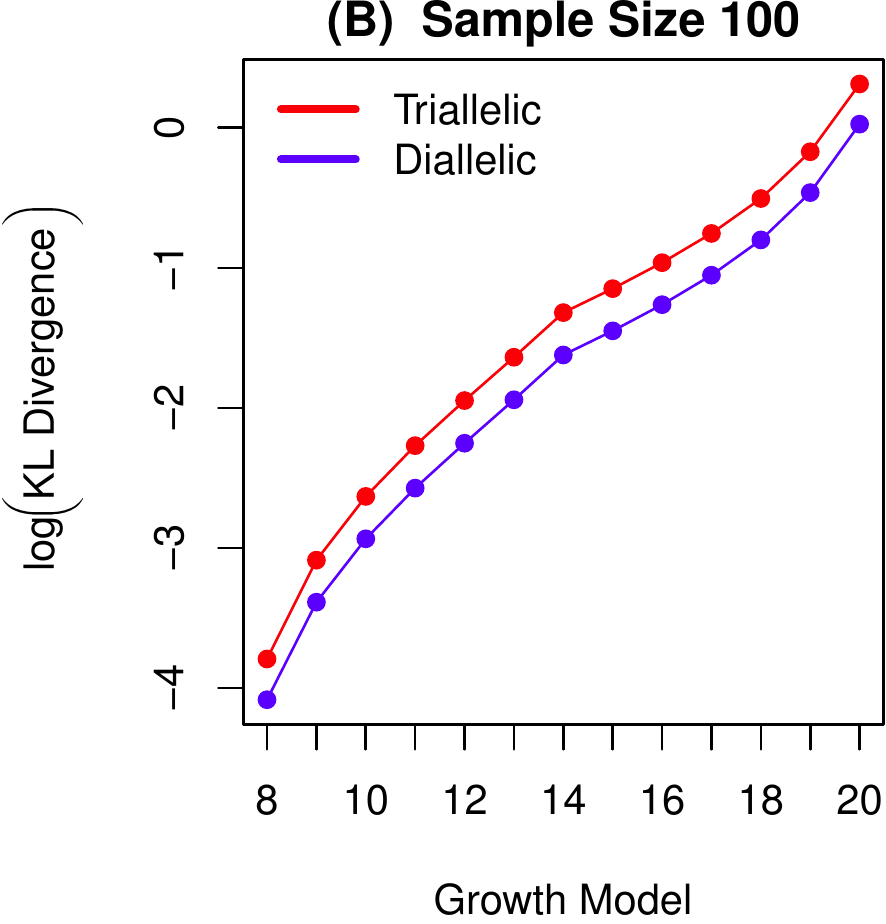}\\
~\\
\includegraphics[width = 0.48 \textwidth]{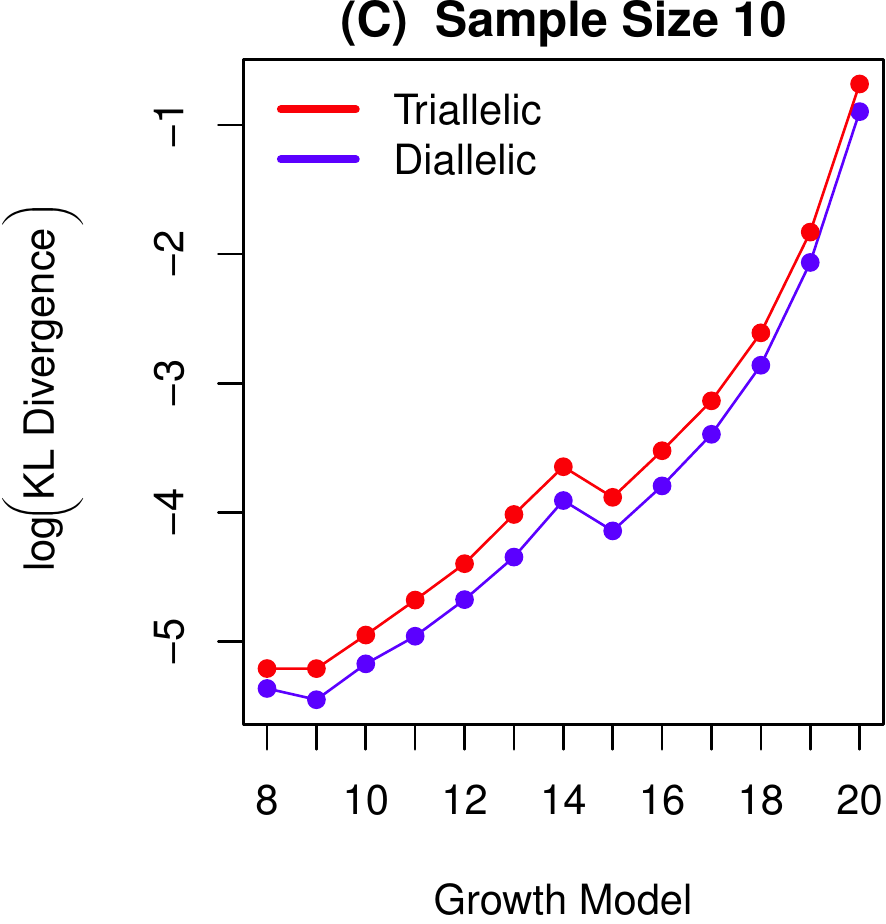}\hfill
\includegraphics[width = 0.48 \textwidth]{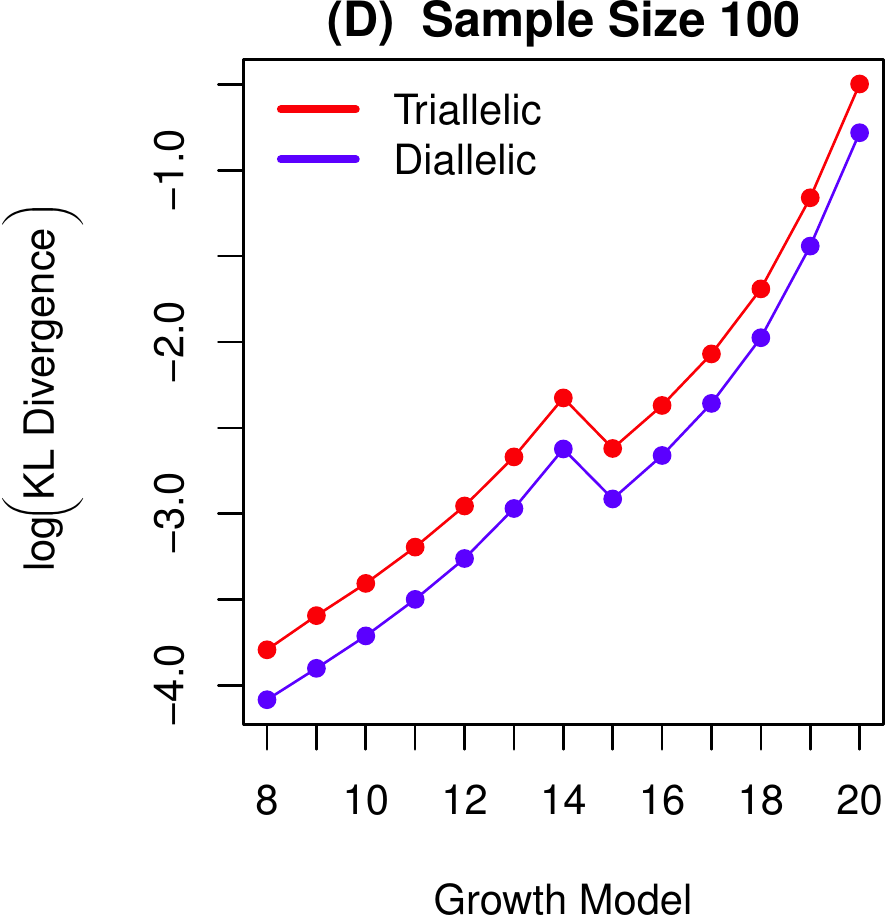}
\caption{\label{fig:suddengrowthfixedtime}KL divergence from one model of instantaneous population growth to another. (A) and (B) depict (base $10$ log) KL divergence $D(\cG_i || \cG_0)$ of the sample frequency spectrum taken from a population under growth model $\cG_i$ from the sample frequency spectrum taken from a population of fixed size ($\cG_0$), for each $i = 8,\ldots 20$. 
(C) and (D) depict (base $10$ log) KL divergence $D(\cG_{i-1} || \cG_i)$ of the sample frequency spectrum taken from a population under growth model $\cG_i$ from the sample frequency spectrum taken from a population under growth model $\cG_{i-1}$, for $i = 8,\ldots, 20$ (where we define the reference growth model before model 8 to simply be a demography with fixed-population size). 
Results are shown for the diallelic (blue) and triallelic (red) sample frequency spectrum. }
\end{figure}

Repeating the steps of our analysis of the exponential growth models, we computed KL divergences $D(\cG_i || \cG_0)$ for $i = 8,\ldots, 29$ (Figures \ref{fig:suddengrowthfixedtime}A, \ref{fig:suddengrowthfixedtime}B, \ref{fig:suddengrowthfixedrate}A, and \ref{fig:suddengrowthfixedrate}B), and to investigate the potential benefit of triallelic spectra in the more subtle problem of distinguishing between two instantaneous population growth models with different degrees of growth, we computed $D(\cG_i || \cG_{i-1})$ for each $i = 9,\ldots, 20$ and $i = 22, \ldots, 29$ (Figures \ref{fig:suddengrowthfixedtime}C, \ref{fig:suddengrowthfixedtime}D, \ref{fig:suddengrowthfixedrate}C and \ref{fig:suddengrowthfixedrate}D).

\begin{figure}[p]
\includegraphics[width = 0.48 \textwidth]{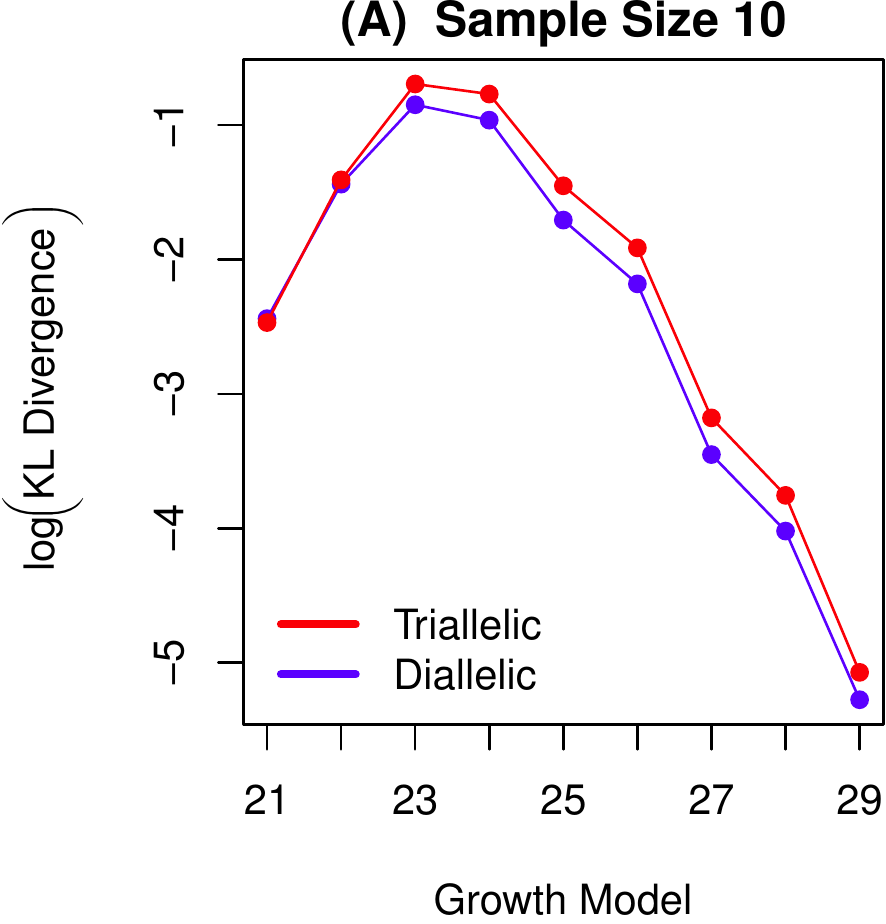}\hfill
\includegraphics[width = 0.48 \textwidth]{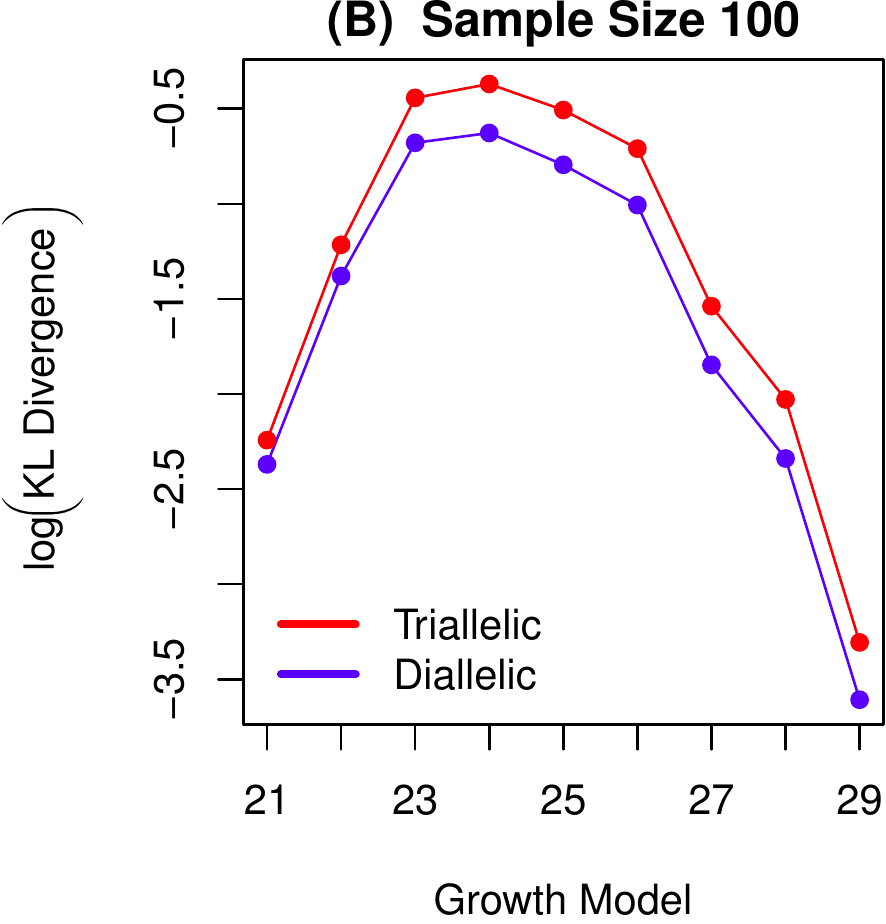}\\
~\\
\includegraphics[width = 0.48 \textwidth]{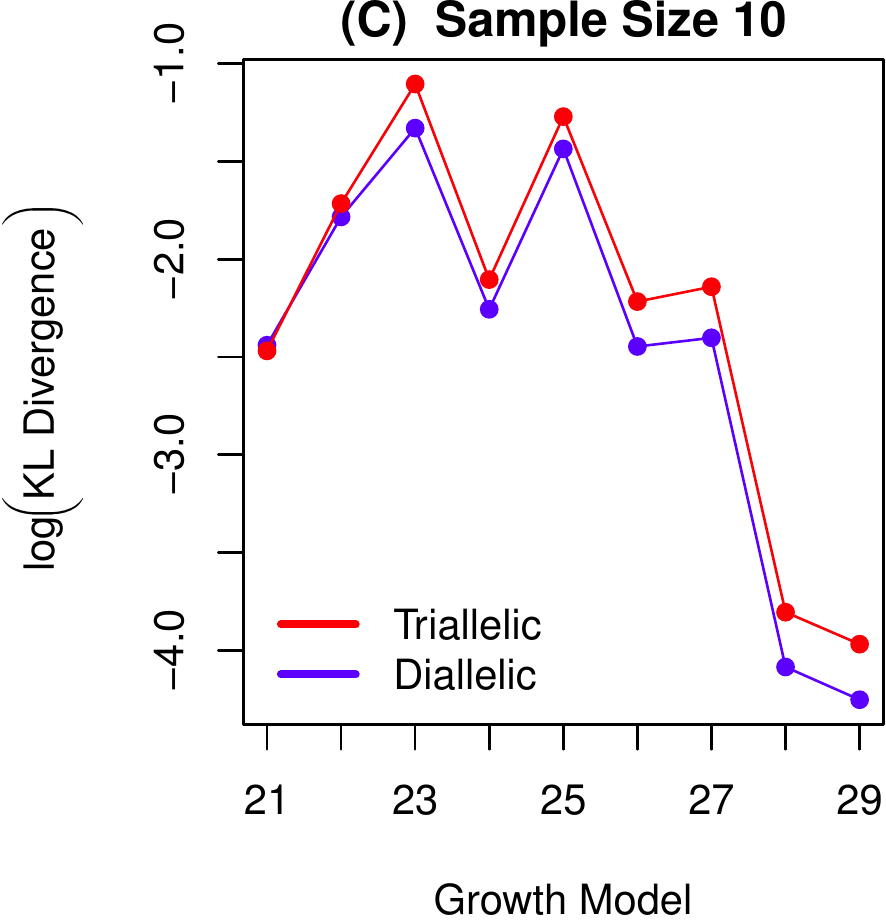}\hfill
\includegraphics[width = 0.48 \textwidth]{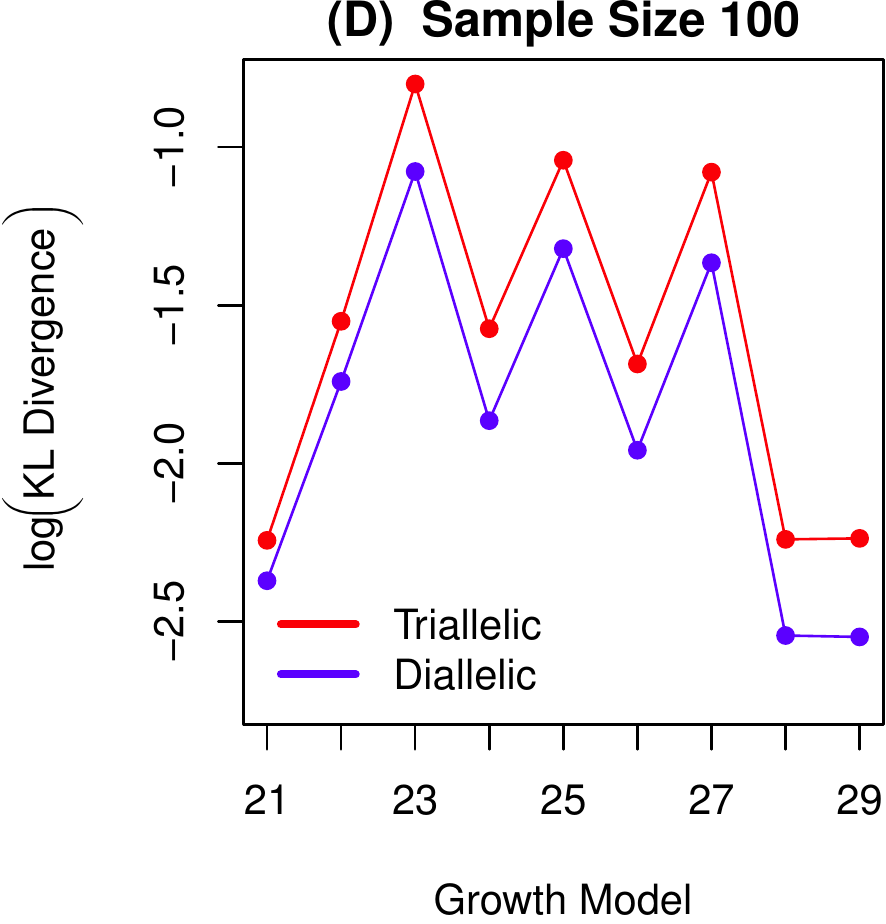}
\caption{\label{fig:suddengrowthfixedrate}KL divergence from one model of instantaneous population growth to another, as in \fref{fig:suddengrowthfixedtime} but this time for models $i = 21,\ldots, 29$.}
\end{figure}

From Figures \ref{fig:suddengrowthfixedtime} and \ref{fig:suddengrowthfixedrate}, we again find that KL divergence from the spectrum under one growth model to another is larger when we use a triallelic, rather than diallelic, sample frequency spectrum, though this advantages fades as the time of the sudden size-change is moved extremely far into the past while the amount of instantaneous growth is kept constant (as in Models 21 and 22 in \fref{fig:suddengrowthfixedrate}). The advantage of the triallelic sample frequency spectrum is even more pronounced for sample size $100$ than for sample size $10$. The mean KL divergence $D(\cG_i || \cG_{i-1})$ (over $i=8,\dots,20$; see \fref{fig:suddengrowthfixedtime}C and \ref{fig:suddengrowthfixedtime}D) is increased by $79\%$ when triallelic spectra are used in place of diallelic spectra for samples of size $10$ (and this divergence is increased by $97\%$ for sample size $100$). For models 21--29 (Figures \ref{fig:suddengrowthfixedrate}C and \ref{fig:suddengrowthfixedrate}D, the mean KL divergence $D(\cG_i || \cG_{i-1})$ is increased by $56\%$ through the inclusion of the third allele for samples of size $10$ (and the divergence is increased by $83\%$ for sample size $100$). Thus, the inclusion of the third allele in triallelic spectra contains information which may considerably increase our power to discern between competing instantaneous growth models with similar parameters.

\section{Application II: Simultaneous mutation model}
\subsection{Theory} As discussed above, \citet{hod:eyr:2010} propose that there exists another mechanism of mutation responsible for the observed excess of triallelic sites in samples of human genomes: the simultaneous generation of two new alleles within a single individual. It is estimated that this mechanism is responsible for the generation of approximately 3\% of all human SNPs. In support of this hypothesis they developed a phylogenetic statistic to test whether the two minor alleles of a triallelic site are closer to each other on a reconstructed phylogenetic tree than would be expected by chance. Using this test, they find significant evidence at the 5\% level for proximity of the minor alleles when probabilities are combined across all triallelic sites in their data, although the null hypothesis of two independent mutation events is rejected at a rate close to the nominal 5\% when each of 113 triallelic sites is tested independently. There are, however, some limitations to this test. First, as \citet{hod:eyr:2010} observe, phylogenies are reconstructed using local haplotype information but ignoring the confounding effects of recombination. Second, it uses the mean branch length between leaf nodes subtended by minor alleles as an indirect measure of the branch length between the minor alleles themselves. Finally, it uses the \emph{proximity} of two mutation events on the phylogeny as evidence for what is in fact a stricter hypothesis of \emph{colocation}. In this section we use our results on the triallelic frequency spectrum to develop a complementary test which does not suffer from these issues. Our approach is to find the sample frequency spectrum as predicted by the simultaneous mutation mechanism and then to compare it with the results of \thmref{thm:main} via a likelihood ratio test. It will be clear that the two spectra give very different predictions, particularly with regard to the expected number of singleton alleles in a sample. 

The key observation which enables us to obtain the frequency spectrum under this model is as follows. Suppose there exists a mechanism whereby two new alleles are produced within a single individual, such as that described in the Introduction: a single DNA duplex within a diploid cell experiences subsequent mutations of both of the nucleotides within a base pair. The duplex then undergoes replication so that two new alleles are produced (\fref{fig:hod:eyr:2010}). Now, the two alleles are observed in a sample taken from the population in the present day. The individual responsible for the creation of these alleles must have been an ancestor of individuals in the sample carrying \emph{either} of the derived alleles. Moreover, this ancestor was the \emph{most recent common ancestor} of any pair of individuals carrying the two distinct derived alleles. The genealogy relating the sample at this site must be of the form illustrated in \fref{fig:simultaneous}; in particular, the simultaneous mutation event coincides with the coalescence node uniting the two clades defined by individuals carrying the two derived alleles. We can therefore condition on this coincidence event in deriving the sample frequency spectrum under this model, by choosing uniformly amongst the $n-1$ coalescence nodes. As noted by \citet{hod:eyr:2010}, the probability that both products of a single human meiosis leave descendants in the following generation is negligible, so the posited simultaneous mutation event is presumed to occur during the mitotic phase of germ-line development.

\begin{figure}[t]
\centering
\includegraphics{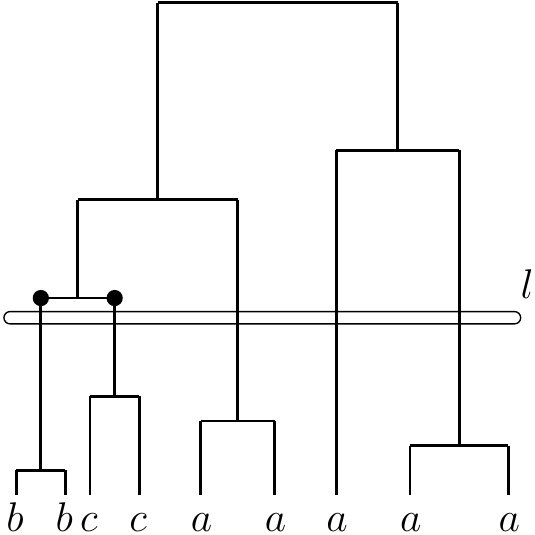}
\caption{\label{fig:simultaneous}A coalescent tree with one simultaneous mutation event. The allele of each leaf is annotated. Also annotated is the variable $l$ (here, $l=5$), the number of lineages ancestral to the sample just prior to the simultaneous mutation event.}
\end{figure}

In its full generality, the mutations under this model are parametrized by a $K \times \binom{K}{2}$ transition matrix $\bfQ = (Q_{i,\{j,k\}})$ whose $(i,\{j,k\})$th entry specifies the probability that a simultaneous mutation affecting allele $i$ gives rise to alleles $j$ and $k$. For notational simplicity we assume $Q_{i,\{j,k\}} = 0$ if $i,j,k$ are not all distinct---it is straightforward to make the appropriate modifications to relax such an assumption. In this setting we have the following theorem.
\begin{theorem}
	\label{thm:sim}
	Let $n_a\bfe_a + n_b\bfe_b + n_c\bfe_c$ denote a triallelic sample, and let $E'_s$ denote the event that there were $s$ instances of the mechanism of simultaneous mutation in the genealogical history relating the sample. Then the sample frequency spectrum is
	\begin{align}
		\phi_\text{S}(n_a,n_b,n_c) &= \bbP(\bfn = n_a\bfe_a + n_b\bfe_b + n_c\bfe_c\mid O_3, E'_1),\notag\\
		&= Q_{a,\{b,c\}}\cdot\frac{2n}{n-2}\cdot\frac{1}{(n-n_a-1)(n-n_a)(n-n_a+1)}.\label{eq:sim}
	\end{align}
\end{theorem}
\begin{proof}
	See the Appendix.
\end{proof}
We remark that in the above we conditioned on $E'_1$; equivalently one could introduce a rate parameter for the occurrences of simultaneous mutations and then let it go to zero after conditioning on $O_3$ only, in which case none of $E'_2$, $E'_3$, $\ldots$ contributes to the frequency spectrum.

Importantly, the frequency spectrum in \eref{eq:sim} depends on the distribution of topologies of coalescent trees but not on the distribution of $\bfT$. Thus, \thmref{thm:sim} continues to apply when we allow a general distribution of inter-coalescence times as in the null model, and in particular this includes a model of variable population size, $N_t$. To summarize: Under a model in which the population size $N_t$ is allowed to vary in time, the sample frequency spectrum of a triallelic site is given by \eref{eq:main} when the two mutation events occur independently and by \eref{eq:sim} when they occur simultaneously within a single individual. An example of the two spectra is shown in \fref{fig:spectrum_comparison}. Clearly, the largest difference occurs in the frequency class of double-singletons, $(n_a,n_b,n_c) = (n-2,1,1)$, which contributes 0.37 of the total probability mass under the simultaneous mutation mechanism compared with 0.09 when the two mutations occur independently along the tree. Other nearby configurations in which both derived alleles are at very low frequency are also overrepresented according to the simultaneous mutation mechanism by comparison with independent mutations, while configurations in which one or both derived alleles are at moderate frequency are slightly underrepresented. This underrepresentation is greatest for frequencies of the form $(n_a,1,n_c)$ and $(n_a,n_b,1)$, i.e. along the axes in \fref{fig:spectrum_comparison}.

\begin{figure}[t]
\centering
({\bf A}) \hfill ({\bf B}) \hfill \phantom{~}
{\includegraphics[width=\textwidth, clip = true, trim=40 5 45 5]{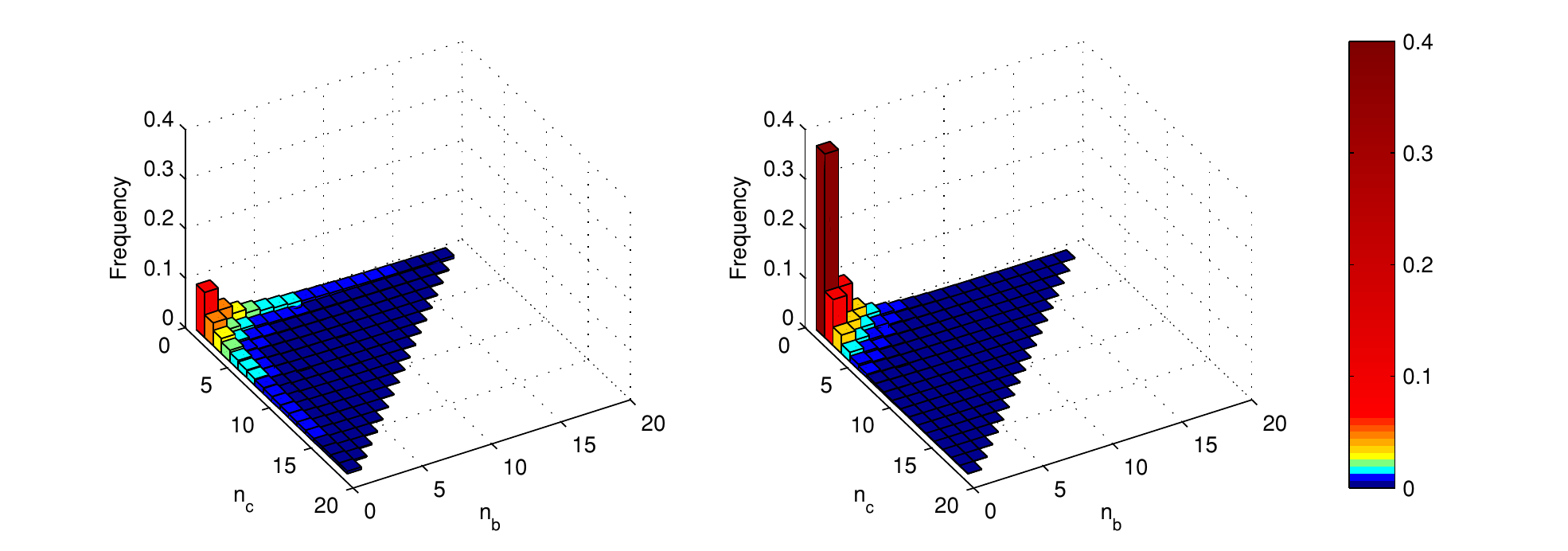}}
\caption{\label{fig:spectrum_comparison}The triallelic sample frequency spectrum when mutations occur ({\bf A}) independently and ({\bf B}) simultaneously, for a sample of size $n = 20$. In this example there are $K=3$ alleles $a$, $b$, and $c$, with uninformative mutation matrices: $P_{ab} = P_{ac} = P_{bc} = P_{cb} = \frac{1}{2}$, and $Q_{a,\{b,c\}} = 1$. The population size is constant in time.}
\end{figure}

\subsection{Likelihood ratio test of independent versus simultaneous mutation} Our goal now is to test for a relative excess of triallelic sites which conform to the frequency spectrum of the simultaneous mutation mechanism. We take as our null hypothesis that each triallelic site was generated by two independent mutation events, so that the frequency spectrum is given by $\phi_0$ [equation \eref{eq:main}]. An appropriate alternative is that some fraction, $\lambda > 0$, of triallelic sites arose as the result of a simultaneous mutation event. Under this model, the sample frequency spectrum is given by the mixture
\begin{equation}
	\label{eq:mixture}
	\phi_\lambda(n_a,n_b,n_c) = (1-\lambda)\phi_0(n_a,n_b,n_c) + \lambda\phi_\text{S}(n_a,n_b,n_c).
\end{equation}
Suppose we observe $M$ triallelic sites, and the $i$th site has configuration $(n^{(i)}_{a_i},n^{(i)}_{b_i},n^{(i)}_{c_i})$. If each pair of sites are sufficiently far apart that their genealogical histories are independent, then a likelihood ratio statistic for these data is
\begin{equation}
	\label{eq:Lambda}
\Lambda = \frac{\prod_{i=1}^M \phi_{\widehat{\lambda}}(n^{(i)}_{a_i},n^{(i)}_{b_i},n^{(i)}_{c_i})}{\prod_{i=1}^M \phi_0(n^{(i)}_{a_i},n^{(i)}_{b_i},n^{(i)}_{c_i})},
\end{equation}
where $\widehat{\lambda}$ is a MLE for $\lambda$, with $0 \leq \widehat{\lambda} \leq 1$. We reject the null hypothesis at level $\alpha$ if $\Lambda$ lies within the $100(1-\alpha)$th percentile tail of its null distribution. Unfortunately, the null distribution of $-2\ln\Lambda$ does not tend toward the usual $\chi^2_1$-distribution due to the possibility that the mixture parameter $\lambda$ lies on the boundary of its permissible set \citep{sel:lia:1987}. We thus employ bootstrap estimation to determine the null distribution of the test statistic by simulation (further described in Supporting Information File S1).

Our aim is to apply this test to empirical SNP data, so we will take $K = 4$ with alleles $\{\text{A, C, G, T}\}$. However, to apply the test we must also specify both transition matrices $\bfP$ and $\bfQ$. While $\bfP$ can be estimated by, for example, using the empirical frequencies of each type of mutation event inferred from diallelic SNPs and their corresponding outgroup alleles \citep{cha:etal:2012}, there is no guidance on how to choose $\bfQ$. In fact, we do not expect the test to be greatly influenced by the \emph{identities} of the three alleles at a site compared to the information contained in the sample counts themselves. Therefore we choose to eliminate the appearance of the entries of $\bfQ$ in $\Lambda$ by \emph{conditioning} on the identities of the observed alleles at each site. Formally, we replace $O_3$ in the definitions above with $O_3^{(a,b,c)}$, the event that three alleles are observed and these alleles are $a$, $b$, and $c$, with $a$ ancestral. This leads to the slightly modified test statistic $\widetilde{\Lambda}$, obtained by replacing the appearances of $\phi_\text{S}$ and $\phi_0$ in each of \eref{eq:mixture} and \eref{eq:Lambda} with
\begin{align*}
	\widetilde{\phi}_\text{S}(n_a,n_b,n_c) &= \bbP(\bfn = n_a\bfe_a + n_b\bfe_b + n_c\bfe_c\mid O_3^{(a,b,c)}, E'_1),\\
	\widetilde{\phi}_0(n_a,n_b,n_c) &= \lim_{\theta\to 0}\bbP(\bfn = n_a\bfe_a + n_b\bfe_b + n_c\bfe_c\mid O_3^{(a,b,c)}).
\end{align*}
By repeating the reasoning that led to expressions for $\phi_\text{S}$ and $\phi_0$, we obtain $\widetilde{\phi}_\text{S}$ from $\phi_\text{S}$ by dividing by $Q_{a,\{b,c\}}$ in \eref{eq:sim}, and $\widetilde{\phi}_0$ from $\phi_0$ by replacing $\kappa_{j,k}$ in \eref{eq:main} with
\[
\widetilde{\kappa}_{j,k} = \left[P_{ab}P_{bc} + P_{ac}P_{cb}\right]D_{j,k} + \left[2P_{ab}P_{ac}\right]G_{j,k}.
\]
The resulting test statistic $\widetilde{\Lambda}$ is independent of $\bfQ$ and is more robust than $\Lambda$ to the choice of $\bfP$. 

\subsection{Data} 
To apply our test to SNP data we first need to specify $\bfP$, and to do this we followed the procedure described in \citet{cha:etal:2012}. Their method requires empirical counts of each type of diallelic SNP with the ancestral allele at each SNP specified, as well as the overall abundance of each nucleotide in the genome. To obtain empirical diallelic SNP data we used the Genome Variation Server (v6.01) \citep{GVS:Aug2011}. We obtained each diallelic SNP from the GVS database, discarding those for which the orthologous chimpanzee allele was unavailable or did not match any of the human alleles. In order to restrict our attention to neutral mutation events occurring at a typical genomic rate, we further discarded SNPs at which one of the alleles would produce a CpG dinucleotide; we discarded SNPs residing in coding regions; and in order to keep our estimate of $\bfP$ independent of the test dataset we discarded sites at which more than two alleles were observed. Assuming that the chimp carries the ancestral allele at each site polymorphic in humans, and that each remaining SNP represents a single mutation event from the ancestral to the derived allele, we were left with the following counts of each type of mutation event:
\begin{align*}
	N_{A\to T} &= 2299 & N_{A\to C} &= 1886 & N_{A\to G} &= 7226\\
	N_{T\to A} &= 2238 & N_{T\to C} &= 6956 & N_{T\to G} &= 1960\\
	N_{C\to A} &= 2819 & N_{C\to T} &= 9940 & N_{C\to G} &= 2395\\
	N_{G\to A} &= 9931 & N_{G\to T} &= 2870 & N_{G\to C} &= 2394. 
\end{align*}
These counts can be converted to rates of mutation by comparison with the overall genomic abundance of each type of nucleotide, which were obtained from the UCSC Genome Browser (hg19) \citep{ken:etal:2002:Aug2011}, after excluding both CpG dinucleotides and the Y chromosome and weighting the counts for the X chromosome by $3/4$. This left the following counts of each type of nucleotide in the human genome: $N_A = 835,878,173$ (30.1\%), $N_T = 836,874,687$ (30.0\%), $N_C = 552,795,868$ (19.9\%), and $N_G = 553,090,147$ (19.9\%). Using these values together with the mutation counts given above, we obtained the following empirical $\bfP$ matrix using the method as described in \citet{cha:etal:2012}: 
\begin{equation*}
\bfP = \;\bordermatrix{\phantom{AA} & \text{A} & \text{T} & \text{C} & \text{G}\cr
                \text{A} & 0.503 & 0.100 & 0.082 & 0.315\cr
                \text{T} & 0.097  & 0.515 & 0.303 & 0.085 \cr
                \text{C} & 0.186 & 0.655 & 0.002 & 0.158\cr
                \text{G} & 0.654  & 0.189 & 0.158 & 0}.
\end{equation*}

We reanalyzed the triallelic dataset of \citet[their Table S1]{hod:eyr:2010}, which comprised 113 triallelic sites. These were in turn obtained from 896 nuclear genes sequenced as part of the Environmental Genome Project \citep{NIEHS:Aug2011} and the SeattleSNPs project \citep{seattleSNPs:Aug2011}, which provide high quality resequencing data, avoiding problems such as ascertainment bias. Only sites of high quality ($Q > 25$), outside CpG dinucleotides, and outside coding regions were included in the data. Orthologous chimpanzee alleles corresponding to each site in the data were kindly provided to us by Alan Hodgkinson; sites for which the chimp allele was unavailable were excluded from our analysis, leaving $M = 96$ triallelic sites. Since these sites originate from different experiments using different population panels, the sample size varied across sites. The minimum, mean, and maximum sample sizes across the 96 triallelic sites were 71, 160, and 190, respectively.

To compute $\Lambda$, the required formulas are given in terms of the joint moments $\bbE[T_jT_k]$, and as before we pre-compute these numerically. This pre-computation step can be reused for different choices of $\lambda$ and across segregating sites with the same sample size, and so it does not add to the computational burden significantly.

The most striking feature of the historical human population size is its recent rapid growth \citep{kei:cla:2012}. To examine this, we used the model $\cG_W$ of \citet{wil:etal:2005} as described above. Computing the MLE $\widehat{\lambda}$ under this demographic model yielded $\widehat{\lambda} = 0$, conflicting with the conclusions of \citet{hod:eyr:2010}. This illustrates the importance of accounting for demographic changes when using frequency spectrum data; had we assumed a population of constant size we would find $\widehat{\lambda} = 0.21$ ($p < 0.001$; see Supporting Information File S1, in which we also examine the robustness of this result to assumptions about $\bfP$ and to potential sequencing errors), in better agreement with \citet{hod:eyr:2010}'s estimate of $\lambda \approx 0.5$. The decrease in the value of our estimated mixture parameter can be explained by the fact that singleton sites---and, similarly, doubly-singleton triallelic sites---are relatively more probable under a null model with population growth than one without. Thus, the abundance of doubly-singleton sites which previously gave rise to an extreme likelihood ratio statistic are now explicable under the null model.

A further potential complication of the data used here is population subdivision with migration between subpopulations. Analytic results for the frequency spectrum under complex demography are unavailable even ignoring the issue of recurrent mutations [but for recent progress on this problem see \citet{che:2012}], and so in Supporting Information File S1 we further investigate this issue by a simulation study. We find $\widehat{\lambda} = 0.16$, though the lack of availability of subpopulation labels with the data renders the result non-significant ($p \approx 0.42$; Supporting Information File S1). 

\section{Discussion}
In this paper we have obtained closed-form expressions for the frequency spectrum of a site experiencing two mutation events, of which triallelic sites are an important example, and allowing for the possibility of a varying historical population size; the results generalize those of \citet{gri:tav:1998} and \citet{jen:son:2011:TPB}. We applied our formulas to the question of the ability of the frequency spectrum to discern between closely-related models of population growth, and to the question of the mechanism of mutation that gives rise to triallelic sites.

Demographic inference from SNP data has thus far relied solely on diallelic sites. As sample sizes in sequencing studies grow with the falling cost of the technology, it is likely that an ever-increasing fraction of segregating sites found will be triallelic. In this article, we have illustrated that the triallelic sampling frequency spectrum is more sensitive to historical population size changes under both exponential and instantaneous growth models, and thus it seems likely that improved estimates of population growth parameters may be obtained by incorporating this growing number of triallelic sites in demographic inference analyses. Furthermore, the increased sensitivity of the triallelic spectrum over the diallelic spectrum to discern between different growth models becomes more exaggerated as sample sizes are increased. While triallelic sites remain relatively rare compared to diallelic ones, the above analysis suggests that they are more valuable per site in distinguishing between a variety of demographic models. We hope that the ability of triallelic sites to fine-tune between competing growth models will be especially useful when looking at recent \emph{super}-exponential growth \citep{kei:cla:2012}, though convenient software for this type of growth is not yet available.

We also found the frequency spectrum under a model in which the two derived alleles of a triallelic site may be generated simultaneously \citep{hod:eyr:2010}, and we have developed a likelihood ratio test for the existence of such a mechanism.  This approach is parametrized by the mixture parameter $\lambda$ which represents the fraction of triallelic sites having arisen as a result of the simultaneous mutation mechanism. Assuming a simple randomly-mating population of constant size we find a MLE of $\widehat{\lambda} = 0.21$ which is significantly non-zero. We show that another explanation for the excess of doubly-singleton triallelic sites is rapid recent population growth, but when we posit a realistic demographic model which includes population subdivision and migration as well as recent growth we find only a minor adjustment to $\widehat{\lambda} = 0.16$, supporting the idea that at least some triallelic sites were generated as the result of simultaneous mutation event. This latter estimate is not however significantly different from 0, a state of affairs we can at least partly attribute to a considerable loss of power---the individuals sampled at a substantial fraction of triallelic sites in our dataset are lacking subpopulation labels. (In particular, none of the triallelic sites of the form $(n^{(i)}-2,1,1)$ remained triallelic after removing from the sample individuals of unknown origin.)

Because of the lack of power associated with the available data we treat these results with caution, and do not rule out the possibility that at least some sites were generated by a mechanism of simultaneous mutation. Indeed, most of the information about a mechanism generating an excessive number of triallelic sites is contained in the absolute \emph{number} of triallelic sites observed in the human genome. We did not use this quantity; instead we \emph{conditioned} on the observed number $M$ of triallelic sites and addressed a slightly different question: Given the excessive number of triallelic sites in the genome, is the frequency spectrum of some fraction of these sites consistent with the two derived alleles being generated simultaneously within a single individual? Even if this proposition is rejected, the question of why there is such an excess of triallelic sites remains. We are hopeful that the coming flood of (subpopulation-labeled) genomic data will enable us to address these questions with much improved power.

One lesson of our work is that rare alleles are as vital when looking at triallelic sites as for diallelic ones; singleton alleles at \emph{diallelic} sites are already the focus of other coalescent-based tests of neutrality \citep{fu:li:1993, ach:2009}. While earlier genotyping projects often chose to exclude sites below a given minor allele frequency, typically 5\%, more recent trends---improved sequencing technologies, larger sample sizes, and an interest in rare variants in their own right \citep{cir:gol:2010, cov:etal:2010,  kei:cla:2012, nel:etal:2012, ten:etal:2012}, should serve to make the appropriate data more readily available. As this trend continues, we expect our results to find further use as recurrent mutations manifest themselves more and more commonly. 

\section*{Acknowledgments}
We thank Christoph Theunert for helpful discussions, and Alan Hodgkinson for providing ancestral allele information to accompany their triallelic site data.  We also thank John Wakeley for helpful suggestions for improving the exposition of this paper.  This research is supported in part by a National Institutes of Health grant R01-GM094402, an Alfred P. Sloan Research Fellowship, and a Packard Fellowship for Science and Engineering.

\appendix
\section*{Appendix}
\begin{proof}[Proof of \eref{eq:gri:tav:1998}.]
	The arguments given here mimic in part those found in \citet{jen:son:2011:TPB}, and we refer the reader there for further details. In that paper we work with \emph{unordered} configurations, $\bfn$; here, the argument is easier to illustrate using \emph{ordered} configurations. We denote a random vector of $n$ alleles consistent with the unordered configuration $\bfn$ by $\bfv_\bfn$; by sampling exchangeability there are $n!/\prod_{i=1}^K n_i!$ equiprobable such vectors. First, denote the event that a single mutation occurred in the genealogical history relating the sample and that it gave rise to a derived allele $b$ by $E_1^{(b)}$. Now, inspection of \fref{fig:tree} tells us that any particular coalescent history with leaf configuration $\bfv_\bfn$ and consistent with $E_1^{(b)}$ must exhibit the following sequence of events going back in time, for some $l_a \in \{1,2,\ldots, n_a\}$:
\setlength{\leftmargini}{8mm}     
\begin{itemize}
	\item a collection of $n_a - l_a$ coalescence events of type $a$ lineages and a collection of $n_b-1$ coalescence events of type $b$ lineages, with the two collections in some interspersed ordering;
	\item a mutation event taking the single remaining type $b$ lineage to a type $a$ lineage;
	\item coalescence of the remaining $l_a + 1$ type $a$ lineages to a most recent common ancestor of the sample.
\end{itemize}
The quantity $l_a$ represents the number of extant lineages ancestral to samples of type $a$ at the time of the single mutation event. Suppose the first coalescence event is between two lineages whose allele at the leaves of the tree is $b$; such a coalescence occurs (amongst all possible coalescence events) with probability
\[
\frac{n_b(n_b-1)}{n(n-1)}.
\]
Continuing in this vein back to the most recent common ancestor, we find that the probability of a given sequence $\mathcal{C}_{\bfv_\bfn,l_a}$ of $n-1$ alleles corresponding to each coalescence event and which is consistent with $\bfv_\bfn$ and with $E_1^{(b)}$ satisfies
\[
\bbP(\mathcal{C}_{\bfv_\bfn,l_a}) = \frac{n_a!(n_a-1)!n_b!(n_b-1)!}{n!(n-1)!}l_a(l_a+1).
\]
For example, in \fref{fig:tree}, $\mathcal{C}_{\bfv_\bfn,l_a} = (b,a,a,b,b,a,a,a)$. In fact, for a fixed $l_a$ this probability is invariant with respect to permutations of the alleles in this sequence [i.e.~with respect to the way the coalescence events are interspersed; \citet{jen:son:2011:TPB}]. There are $\binom{n-l_a-1}{n_b-1}$ ways to intersperse the first $n-l_a-1$ coalescence events \citep[Lemma 3.1]{jen:son:2011:TPB}. 

Now, given $\mathcal{C}_{\bfv_\bfn,l_a}$, we require the probability that a single mutation event occurs on the correct branch of the tree, an event we denote by $M_{\bfv_\bfn}$. Since mutation events occur along the branches as a Poisson process of rate $\theta/2$, we have that
\[
\bbP(M_{\bfv_\bfn}\mid\mathcal{C}_{\bfv_\bfn,l_a}) = \bbE\left[\frac{\theta}{2}L_ne^{-\frac{\theta}{2}L_n}\cdot\frac{T_{l_a+1}}{L_n}\right] = \frac{\theta}{2}\bbE[T_{l_a+1}] + O(\theta^2),
\]
as $\theta \to 0$, where $L_n = \sum_{j=2}^n jT_j$ is the total branch length. 

 Putting all this together,
\begin{align}
	\bbP(\bfn = n_a\bfe_a + n_b\bfe_b, E_1^{(b)}) &= \binom{n}{n_b}\bbP(\bfv_\bfn, E_1^{(b)}),\notag\\
&=P_{ab}\binom{n}{n_b}\sum_{l_a=1}^{n_a}  \binom{n-l_a-1}{n_b-1} \bbP(\mathcal{C}_{\bfv_\bfn,l_a})\bbP(M_{\bfv_\bfn}\mid\mathcal{C}_{\bfv_\bfn,l_a}),\notag\\
&= \frac{\theta P_{ab}}{2}\sum_{l_a=1}^{n_a}\binom{n_a-1}{l_a-1}\binom{n-1}{l_a}^{-1}\bbE[T_{l_a+1}] + O(\theta^2). \label{eq:onemutationproof}
\end{align}
Summing over $n_a$, we find
\[
\bbP(E_1^{(b)}) = \sum_{n_a=1}^{n-1}\bbP(\bfn = n_a\bfe_a + n_b\bfe_b, E_1) = \frac{\theta P_{ab}}{2}\sum_{l_a=1}^{n-1}(l_a+1)\bbE[T_{l_a+1}]+ O(\theta^2),
\]
and hence, noting that
\[
\phi(i) := \lim_{\theta\to 0}\bbP[\bfn = (n-i)\bfe_a + i\bfe_b\mid E_1^{(b)}]  = \lim_{\theta\to 0}\frac{\bbP[\bfn = (n-i)\bfe_a + i\bfe_b,E_1^{(b)}]}{\bbP(E_1^{(b)})},
\]
substituting for the numerator and denominator, and letting $\theta \to 0$ yields the given result.
\end{proof}
Henceforward we denote the trinomial coefficient by $\binom{n}{i,j,k}$.
\begin{proof}[Proof of \lemmaref{lem:probs}.]
The argument here is very similar to that of the proof of \eref{eq:gri:tav:1998}. This time a coalescent tree compatible with $\bfv_\bfn$ and with $\EIIn^{(b,c)}$ must exhibit the following sequence of events \citep[Lemma 4.1]{jen:son:2011:TPB}:
\setlength{\leftmargini}{8mm}   
\begin{itemize}
\item $n_a - l_y$ coalescence events of type $a$ alleles, $n_b-m$ coalescence events of type $b$ alleles, and $n_c-1$ coalescence events of type $c$ alleles, in some interspersed ordering;
\item a mutation event reverting the sole remaining type $c$ allele to a type $b$;
\item $l_y-l_o$ coalescence events of type $a$ alleles and $m$ coalescence events of type $b$ alleles;
\item a mutation event reverting the sole remaining type $b$ allele to a type $a$;
\item	coalescence of the remaining $l_o + 1$ type $a$ lineages to a most recent common ancestor of the sample.
\end{itemize} 
Here, $l_y$ is the number of extant type $a$ lineages at the time of the younger mutation event, $l_o$ is the number of type $a$ lineages at the time of the older mutation event, and $m$ is the number of type $b$ lineages at the time of the younger mutation event (\fref{fig:trees}). Arguing as above, any one of the $\binom{n-l_y-m-1}{n_a-l_y, n_b-m, n_c-1}\binom{m+l_y-l_o}{m}$ compatible sequences $\mathcal{C}_{\bfv_\bfn,m,l_y,l_o}$ satisfies
\[
\bbP(\mathcal{C}_{\bfv_\bfn,m,l_y,l_o}) = \frac{n_a!(n_a-1)!n_b!(n_b-1)!n_c!(n_c-1)!}{n!(n-1)!}m(m+1)l_o(l_o+1),
\]
and the two mutation events land on the correct pair of branches with probability
\[
\bbP(M_{\bfv_\bfn}\mid\mathcal{C}_{\bfv_\bfn,m,l_y,l_o}) = \bbE\left[\left(\frac{\theta}{2}L_n\right)^2\frac{e^{-\frac{\theta}{2}L_n}}{2!}\cdot 2\frac{T_{m+l_y+1}}{L_n}\frac{T_{l_o+1}}{L_n}\right] = \frac{\theta^2}{4}\bbE[T_{m+l_y+1}T_{l_o+1}] + O(\theta^3).
\]
Hence, following the reasoning that led to \eref{eq:onemutationproof}, we find
\begin{multline*}
\bbP(\bfn = n_a\bfe_a + n_b\bfe_b + n_c\bfe_c, \EIIn^{(b,c)}) =
P_{ab}P_{bc}\binom{n}{n_a,n_b,n_c}\\
\times\sum_{l_y=1}^{n_a}\sum_{l_o=1}^{l_y}\sum_{m=1}^{n_a}\binom{n-l_y-m-1}{n_a-l_y,n_b-m,n_c-1}\binom{m+l_y-l_o}{m}\bbP(\mathcal{C}_{\bfv_\bfn,m,l_y,l_o})\bbP(M_{\bfv_\bfn}\mid\mathcal{C}_{\bfv_\bfn,m,l_y,l_o}).
\end{multline*}
Substituting for each term in the summand and simplifying leads to \eref{eq:ntwonested}, and summing over triallelic configurations $\bfn$ yields \eref{eq:twonested}.

The nonnested case is similar. There are $\binom{n-m-l_y-1}{n_a-l_y,n_b-m,n_c-1}\binom{m+l_y-l_o}{m-1}$ possible $\mathcal{C}_{\bfv_\bfn,m,l_y,l_o}$ \citep[Lemma 4.2]{jen:son:2011:TPB}, each with probability
\[
\bbP(\mathcal{C}_{\bfv_\bfn,m,l_y,l_o}) = \frac{n_a!(n_a-1)!n_b!(n_b-1)!n_c!(n_c-1)!}{n!(n-1)!}l_y(l_y+1)l_o(l_o+1).
\]
The two mutation events land on the correct pair of branches and are in the correct age-order with probability
\begin{align}
\bbP(M_{\bfv_\bfn}^{(b,c)}\mid\mathcal{C}_{\bfv_\bfn,m,l_y,l_o}) &= \bbE\left[\left(\frac{\theta}{2}L_n\right)^2\frac{e^{-\frac{\theta}{2}L_n}}{2!}\cdot 2\frac{T_{m+l_y+1}}{L_n}\frac{T_{l_o+1}}{L_n}\right]\cdot\frac{1}{1+\delta_{m+l_y,l_o}},\label{eq:age-order}\\
&= \frac{\theta^2}{4}\frac{\bbE[T_{m+l_y+1}T_{l_o+1}]}{1+\delta_{m+l_y,l_o}} + O(\theta^3).\notag
\end{align}
The additional factor on the right-hand side of \eref{eq:age-order} accounts for the fact that, should the two mutations arise during the same epoch ($m+l_y+1 = l_o + 1$), then only with probability $1/2$ is the mutation giving rise to the $b$ allele the elder one. Finally,
\begin{multline*}
\bbP(\bfn = n_a\bfe_a + n_b\bfe_b + n_c\bfe_c, \EIInn^{(b,c)}) =P_{ab}P_{ac}\binom{n}{n_a,n_b,n_c}\\
\times\sum_{l_y=1}^{n_a}\sum_{l_o=1}^{l_y+1}\sum_{m=1}^{n_b}\binom{n-m-l_y-1}{n_a-l_y,n_b-m,n_c-1}\binom{m+l_y-l_o}{m-1} \bbP(\mathcal{C}_{\bfv_\bfn,m,l_y,l_o})\bbP(M^{(b,c)}_{\bfv_\bfn}\mid\mathcal{C}_{\bfv_\bfn,m,l_y,l_o}),
\end{multline*}
which yields \eref{eq:ntwononnested} after substituting and simplifying. Summing over triallelic configurations $\bfn$ yields \eref{eq:twononnested}.
\end{proof}
\begin{proof}[Proof of \thmref{thm:sim}.]
	Conditional on the event $E'_1$, the single simultaneous mutation event occurs uniformly on the $n-1$ coalescence nodes of the coalescent tree. Hence, immediately prior to (i.e.~more recently than) the coalescence event at which the mutation event occurred, there are $L$ lineages ancestral to the present-day sample, with $L\sim\text{Uniform}\{2,\ldots,n\}$. Given $L=l$, it is well known \citep[e.g.][]{gri:tav:1998} that the distribution of the number of leaves subtending each of these $l$ lineages is uniform on the $\binom{n-1}{l-1}$ possible compositions of $n$. Of these, $\binom{n_a-1}{l-3}$ compositions have $n_b$ leaves and $n_c$ leaves respectively subtending the two lineages coalescing into the node that experiences the simultaneous mutation event. Hence,
	\[
	\bbP(\bfn = n_a\bfe_a + n_b\bfe_b + n_c\bfe_c \mid E'_1, \{L=l\}) = Q_{a,\{b,c\}}\frac{\binom{n_a-1}{l-3}}{\binom{n-1}{l-1}},
	\]
for $3\leq l \leq n$ (and is $0$ otherwise). 
Also note that
\[
\bbP(O_3\mid E_1') = \frac{n-2}{n-1}.
\]
This is the probability that the simultaneous mutation event does not occur at the oldest of the coalescence nodes ($l = 2$), which would lead to a sample containing no copies of the ancestral allele.

Putting all this together we have
\begin{align*}
	\bbP(\bfn = n_a\bfe_a + n_b\bfe_b + n_c\bfe_c \mid O_3, E'_1) &= \frac{\bbP(\bfn = n_a\bfe_a + n_b\bfe_b + n_c\bfe_c , O_3\mid E'_1)}{\bbP(O_3 \mid E'_1)},\\
	&= \frac{\sum_{l=2}^n \bbP(L = l\mid E'_1)\bbP(\bfn = n_a\bfe_a + n_b\bfe_b + n_c\bfe_c ,O_3\mid E'_1,\{L=l\})}{\bbP(O_3 \mid E'_1)},\\
	&= \frac{n-1}{n-2}\sum_{l=3}^n \frac{1}{n-1}\cdot Q_{a,\{b,c\}}\binom{n_a-1}{l-3}\binom{n-1}{l-1}^{-1},
\end{align*}
which simplifies to equation \eref{eq:sim}.
\end{proof}

\newpage
\bibliographystyle{myplainnat}
\bibliography{master}
\newpage
\setcounter{page}{1}               
\pagestyle{fancy}                        
\renewcommand{\headrulewidth}{0pt} 
\renewcommand{\footrulewidth}{0pt}      
\renewcommand{\thefigure}{S\arabic{figure}}
\renewcommand{\thetable}{S\arabic{table}}     
\renewcommand{\thesection}{\arabic{section}}
\setcounter{figure}{0}
\setcounter{table}{0}
\fancyhead{}  
\fancyfoot[LR]{}
\fancyfoot[C]{\thepage\ SI} 
\begin{center}
{\Large \bf Supporting Information}

\vspace{5mm}                                 
\begin{spacing}{1.5}
{\Large \bf General triallelic frequency spectrum under demographic models with variable population size}
\end{spacing}

\vspace{0.5cm}
%
%
%
%
{Paul A. Jenkins,\hspace{5mm} Jonas W. Mueller, \hspace{5mm} Yun S. Song}
\end{center}

\section{Tests for the existence of a mechanism of simultaneous mutation}

In the main text we describe an analysis of the simultaneous mutation mechanism under a population model in which there is random mating and for which the population size varies in time. In this supplement we investigate the effect of varying our modelling assumptions on the conclusions of this test.

\subsection{Test for a population of constant size.} Consider a simpler model in which $N_t \equiv N_0$. We applied our likelihood ratio test to the data described in the main text, and obtained $\ln\widetilde{\Lambda} = 5.71$ and $\widehat{\lambda} = 0.21$. Because we have obtained closed-form expressions for $\widetilde{\phi}_0$, $\widetilde{\phi}_\text{S}$ and therefore also $\widetilde{\phi}_\lambda$, it is straightforward to compute the whole likelihood curve for $\lambda$ (\fref{fig:likelihood}). We estimated a 95\% confidence interval of $[0.07, 0.35]$ for $\lambda$ by parametric bootstrap. In other words, we simulated a complete dataset of $M = 96$ triallelic sites with the three observed alleles $a_i$, $b_i$, $c_i$, and the sample size $n^{(i)}$, fixed to match their observed values at the $i$th site, for each $i=1,\ldots,M$. Sample counts were then drawn from $\widetilde{\phi}_{\widehat{\lambda}}$ for each site, and a bootstrapped MLE, $\widehat{\lambda}^*$, was computed for this simulated sample. The entire procedure was repeated 1000 times to obtain an empirical distribution $\{\widehat{\lambda} - \widehat{\lambda}_k^*\}_{k=1,\ldots,1000}$, which serves as an approximation of the distribution of $\lambda - \widehat{\lambda}^*$. The confidence interval is obtained from the central 95th percentile of this distribution (centered around $\widehat{\lambda}$) (\fref{fig:likelihood}). The upper limit of our confidence interval for $\widehat{\lambda}$ is slightly lower than an estimate of $\lambda \approx 0.5$ by \citet{hod:eyr:2010}, who used a different method but also assumed a population of constant size in their calculation. 

We further estimated the null distribution of $\widetilde{\Lambda}$ by bootstrap simulation; because the parameter of interest is not necessarily an interior point of the interval $[0,1]$, the usual $\chi^2_1$ approximation does not apply. We used the same procedure for the simulation of bootstrap samples as described above, except this time drawing sample counts from $\widetilde{\phi}_0$ rather than $\widetilde{\phi}_{\widehat{\lambda}}$. This method yielded an empirical distribution of test statistics, $\{\ln\widetilde{\Lambda}_k\}_{k=1,\ldots,1000}$ (\fref{fig:Lambda}), whose 95th\% percentile tail 
provides an appropriate $p$-value. In fact, the largest empirical $\widetilde{\Lambda}$ value was 4.03, much less than our observed value. From this we conclude that $p < 0.001$ and we reject the null hypothesis that $\lambda = 0$. This is consistent with $\lambda = 0$ lying outside our 95\% confidence interval for $\widehat{\lambda}$.

One important point to note is that nine of the 96 sites in the data overwhelmingly contribute to the significance of the test statistic. These nine sites are the SNPs for which $n_b = n_c = 1$. 
Using our formulas for the sampling frequency spectrum under the two possible mutation mechanisms, we have $\widetilde{\phi}_\text{S}(n^{(i)}-2,1,1) > 0.33$ and $\widetilde{\phi}_0(n^{(i)}-2,1,1) < 0.04$ across all sites $i$ with this configuration. 
Thus the sample counts observed at these sites are roughly more than 10 times more likely under the simultaneous mutation mechanism.  However, it could be argued that these highly influential doubly-singleton triallelic sites are also prime candidates for being the result of sequencing error. To examine the dependence of the conclusion of our test on these doubly-singleton sites, we did the following. We first ranked the 9 sites from most influential to least influential, based on the ratio
\[
\frac{\widetilde{\phi}_\text{S}(n^{(i)}-2,1,1)}{\widetilde{\phi}_0(n^{(i)}-2,1,1)},
\]
(where $n^{(i)}$ is the sample size of the $i$th site). We then removed the topmost influential sites in succession and re-computed the test statistics via the previous methodology. Results are shown in \tref{tab:doublysingletons}.

From the table we can see that for our test statistic to lose its significance at the 5\% level, we must remove six or more of the most influential sites. Is it reasonable that at least six of these nine sites are the result of sequencing errors? One of the attractive aspects of the SeattleSNPs project and Environmental Genome Project datasets is that they are resequenced to a high quality, with each SNP confirmed in multiple reactions. Assuming a quality score of $Q > 25$, the probability of a base-calling error in two independent reactions is $10^{-5}$ \citep{hod:eyr:2010}, and it therefore seems reasonable to conclude that it is very unlikely that the presence of more than five of these nine sites are the result of sequencing error. Thus, we reject the null hypothesis that the mixture parameter $\lambda = 0$ even if we allow for at least some of these sites to have resulted from sequencing errors.

In order to test the robustness of our results to the choice of $\bfP$, we also redid our analysis with a ``noninformative'' choice of $\bfP$ given by
\begin{equation*}
\bfP = \;\bordermatrix{\phantom{AA} & \text{A} & \text{T} & \text{C} & \text{G}\cr
                \text{A} & 0 & \frac{1}{3} & \frac{1}{3} & \frac{1}{3}\cr
                \text{T} & \frac{1}{3}  & 0 & \frac{1}{3} & \frac{1}{3} \cr
                \text{C} & \frac{1}{3} & \frac{1}{3} & 0 & \frac{1}{3}\cr
                \text{G} & \frac{1}{3}  & \frac{1}{3} & \frac{1}{3} & 0}.
\end{equation*}
This time we find $\widehat{\lambda} = 0.22$ and $\ln\widetilde{\Lambda} = 7.02 > 6.15 = \max_k\ln\widetilde{\Lambda}^*_k$, and thus $p < 0.001$ as before. We conclude that our test exhibits good robustness to the perturbations in the exact specification of $\bfP$, at least in the direction of the alternative choice of $\bfP$ investigated here.

\subsection{Test for a subdivided population of nonconstant size with migration.} To extend our test for the mechanism that gives rise to triallelic sites to a subdivided population with migration, we proceed by simulation. There are 
two complications which make this approach more difficult. First, the sample space for the joint frequency spectrum is over collections of vectors, where each vector in the collection records the allele counts in one of the subpopulations. This vastly increases the size of the space over which simulation must be performed, dramatically increasing computation time. Second, to compute the frequency spectrum under this model requires all of our samples to be subpopulation-labelled. As we describe below, this is far from the case for data from SeattleSNPs and NIEHS. Discarding samples whose subpopulation membership is unknown reduces the sample size and hence the power of any test.

We now describe our Monte Carlo procedure for approximating the joint frequency spectrum of (subpopulation-labelled) samples from a subdivided population. It is based on an algorithm of \citet{hud:2001:G} who was interested in two linked diallelic SNPs rather than a single triallelic SNP; we give a brief description and refer the reader there for further details. First consider generalizing $\widetilde{\phi}_0(n_a,n_b,n_c)$ to a subdivided population, which we write in the form $\widetilde{\phi}_0(\bfm)$ with $\bfm = (n_a^{[u]},n_b^{[u]},n_c^{[u]})_{u=1,\ldots,U}$ denoting the collection of sample counts for each subpopulation, $u=1,\ldots,U$. (We focus on a single site, and the superscript refers to the subpopulation label rather than a site index. The site is assumed to be triallelic in the combined sample, but it need not be triallelic within each subpopulation.) In order to take a Monte Carlo approximation we first write the relevant sampling probabilities in the following form:
\begin{align*}
	\bbP(\bfm, O_3^{(a,b,c)}) &= \bbE\left[\sum_{j,k}\bbI\{\bfm,\mathbf{\tau},j,k\}(1-e^{-\theta l_j/2})(1-e^{-\theta l_k/2})e^{-\theta(L_n - l_j - l_k)/2}P_{a(j)b}P_{a(k)c}\right],\\
	&= \frac{\theta^2}{4}\bbE\left[\sum_{j,k}\bbI\{\bfm,\mathbf{\tau},j,k\}l_jl_kP_{a(j)b}P_{a(k)c}\right] + O(\theta^3).
\end{align*}
Here the expectation is taken with respect to random genealogies drawn according to the given complex demographic history. The function $\bbI\{\bfm,\mathbf{\tau},j,k\}$ is an indicator that ensures we count only those genealogies for which a mutation on branch $j$ from allele $a(j)\mapsto b$ and a mutation on branch $k$ from allele $a(k)\mapsto c$ gives rise to the configuration $\bfm$. The notation $a(\cdot)$ is taken to mean ``the allele at the ancestral node of this branch''. The summation is a double-summation over all branches of $\mathbf{\tau}$, which considers each possible placement of the mutations giving rise to allele $b$ and allele $c$. 
Branch $j$ has length $l_j$, branch $k$ has length $l_k$, and the total tree length is $L_n$. Ignoring terms of $O(\theta^3)$ discards the possibility of more than two mutations, as usual. 

Similarly, we can write
\[
	\bbP(O_3^{(a,b,c)}) = \frac{\theta^2}{4}\bbE\left[\sum_{j,k}\bbI\{O_3^{(a,b,c)},\mathbf{\tau},j,k\}l_jl_kP_{a(j)b}P_{a(k)c}\right] + O(\theta^3),
\]
where $\bbI\{O_3^{(a,b,c)},\mathbf{\epsilon},j,k\}$ is an indicator for the event that a mutation on branch $j$ from $a(j)\mapsto b$ and a mutation on branch $k$ from $a(k)\mapsto c$ gives rise to a triallelic sample comprising the three alleles $a$, $b$, and $c$. Finally, we have that
\begin{align*}
	\widetilde{\phi}_0(\bfm) = \lim_{\theta\to 0}\bbP(\bfm|O_3^{(a,b,c)}) &= \lim_{\theta\to 0}\frac{\bbP(\bfm,O_3^{(a,b,c)})}{\bbP(O_3^{(a,b,c)})},\\
	&= \frac{\bbE\left[\sum_{j,k}\bbI\{\bfm,\mathbf{\tau},j,k\}l_jl_kP_{a(j)b}P_{a(k)c}\right]}{\bbE\left[\sum_{j,k}\bbI\{O_3^{(a,b,c)},\mathbf{\tau},j,k\}l_jl_kP_{a(j)b}P_{a(k)c}\right]},\\
	&\approx \frac{\frac{1}{\cN}\sum_{i=1}^\cN\left[\sum_{j,k}\bbI\{\bfm,\mathbf{\tau}^{(i)},j,k\}l^{(i)}_jl^{(i)}_kP_{a(j)b}P_{a(k)c}\right]}{\frac{1}{\cN}\sum_{i=1}^\cN\left[\sum_{j,k}\bbI\{\bfm,\mathbf{\tau}^{(i)},j,k\}l^{(i)}_jl^{(i)}_kP_{a(j)b}P_{a(k)c}\right]}.
\end{align*}
The final step replaces each of the numerator and denominator with a Monte Carlo sample of $\cN$ simulated genealogies, $\mathbf{\tau}^{(1)},\ldots,\mathbf{\tau}^{(\cN)}$ (which can be reused in both numerator in denominator, though this may introduce some bias). 

The frequency spectrum under the simultaneous mutation mechanism can be obtained in a similar manner. We write
\[
	\bbP(\bfm, O_3^{(a,b,c)}|E_1') = \bbE\left[\sum_{v\in V_\tau} \frac{1}{2}[\bbI\{\bfm,\mathbf{\tau},v,b,c\} + \bbI\{\bfm,\mathbf{\tau},v,c,b\}]\frac{1}{n-1}Q_{a,\{b,c\}}\right],
\]
summing over all coalescence vertices $V_\tau$ of $\mathbf{\tau}$. 
Then vertex $v$ is chosen with probability $1/(n-1)$, and the simultaneous mutation at $v$ produces the correct derived alleles with probability $Q_{a,\{b,c\}}$. The factor of $1/2$ accounts for the two possible assignments of the alleles $b$ and $c$ to the branches descending from the simultaneous mutation event, and there is an indicator corresponding to each of these assignments, which is 1 if the configuration at the leaves of the genealogical tree is $\bfm$ and zero otherwise. Finally, we obtain a Monte Carlo approximation of $\widetilde{\phi}_\text{S}$ by noting that placing the simultaneous mutation event at the root vertex does not give rise to a triallelic site, so
\[
\bbP(O_3^{(a,b,c)}|E_1') = \frac{n-2}{n-1}Q_{a,\{b,c\}},
\]
(see the proof of Theorem 4
), and hence
\begin{align*}
	\widetilde{\phi}_\text{S}(\bfm) 
	&=  \frac{\bbP(\bfm,O_3^{(a,b,c)}|E_1')}{\bbP(O_3^{(a,b,c)}|E_1')}
	= \frac{1}{n-2}\bbE\left[\sum_{v\in V_\tau} \frac{1}{2}[\bbI\{\bfm,\mathbf{\tau},v,b,c\} + \bbI\{\bfm,\mathbf{\tau},v,c,b\}]\right],\\
	&\approx \frac{1}{\cN}\sum_{i=1}^\cN \frac{1}{n-2}\sum_{v\in V_{\tau^{(i)}}}\frac{1}{2}[\bbI\{\bfm,\mathbf{\tau}^{(i)},v,b,c\} + \bbI\{\bfm,\mathbf{\tau}^{(i)},v,c,b\}].
\end{align*}

The above method is extremely general, in the sense that it requires only that one is able to simulate genealogical trees from the assumed model. Indeed, \citet{hud:2001:G} did not consider population subdivision, yet, other than the simulation of genealogical histories from \texttt{ms}, the details of the algorithm are unchanged when subdivision is also considered.

To account for population subdivision, we assumed as the true demographic model that inferred by \citet{gut:etal:2009}. Their method extended that of \citet{wil:etal:2005} to allow for an ancestral population to split into modern subpopulations with rare but continuous migration between subpopulations and each subpopulation allowed to change in size with time. Using data from three HapMap populations: 12 Yoruba individuals from Ibadan, Nigeria (YRI); 22 CEPH Utah residents with ancestry from northern and western Europe (CEU); and 12 Han Chinese individuals from Beijing, China (CHB); \citet{gut:etal:2009} fit a model governing the historical relationship of the three populations during the human expansion out of Africa. The data from SeattleSNPs and the Environmental Genome Project comprise samples from panels made up of individuals in HapMap populations and samples from other panels. Individuals in the latter are classified as of European, African-American, Asian, or Hispanic descent. For simplicity and following \citet{hod:eyr:2010}, we focus on a model only for an African and a non-African (European) population, excluding Asian and Hispanic samples from our analysis. To obtain subpopulation labels \citep[not provided in][]{hod:eyr:2010} we reanalyzed the data available from SeattleSNPs and the Environmental Genome Project, excluding sites in coding regions and so on as described above. This time we excluded from each site any individuals for whom subpopulation labels were not given, leaving $M = 34$ triallelic sites in total.

We applied the Monte Carlo method described above for approximating the frequency spectrum under the null and alternative models of mutation, given the complex demographic history of \citet{gut:etal:2009}. The Monte Carlo error at sites whose sample configurations were of very low probability could unduly affect the overall likelihood ratio statistic, so we varied the Monte Carlo sample size $\cN$, from $10^6$ to $10^{10}$ for different sites, according to a preliminary estimate of the magnitude of the probability of their sample configuration. (As an extreme example, sites with configurations of very low probability could potentially have this probability estimated as 0, which would also cause the overall likelihood ratio statistic to be 0.)
Using this approach we obtained $\ln\widetilde{\Lambda} = 2.01$ ($p \approx 0.42$) and $\widehat{\lambda} = 0.16$. Due to the computationally-intensive nature of the simulation, the $p$-value is based on only a very coarse approximation of the null distribution of $\ln\widehat{\Lambda}$ using 50 bootstrapped test statistics (\fref{fig:Lambda-subdivision}); nonetheless, it is clearly non-significant at the 5\% level.

%
\section{Supporting Figures and Table}

\begin{table}[h]
    \caption{The effect on the test statistic of triallelic sites for which both derived alleles are singletons. As more and more of these sites are removed from the data in succession, the mixture parameter MLE, $\widehat{\lambda}$, and the significance of the test statistic, gradually decrease.}
\label{tab:doublysingletons}
\centering
{
\begin{tabular}{ccc}\hline
Number of sites\\[0pt]
 removed & $\widehat{\lambda}$ & $p$-value\\ \hline
0 & 0.204 & $\leq$0.001 \\
1 & 0.189 & $\leq$0.001 \\
2 & 0.171 & $\leq$0.003 \\
3 & 0.153 & $\leq$0.006 \\
4 & 0.133 & $\leq$0.013 \\
5 & 0.113 & $\leq$0.036 \\
6 & 0.090 & $\leq$0.090 \\
7 & 0.064 & $\leq$0.161 \\
8 & 0.036 & $\leq$0.283 \\
9 & 0.002 & $\leq$0.461 \\
\hline
\end{tabular}
}
\end{table}

\begin{figure}[h]
\centering
\includegraphics[width=0.85\textwidth]{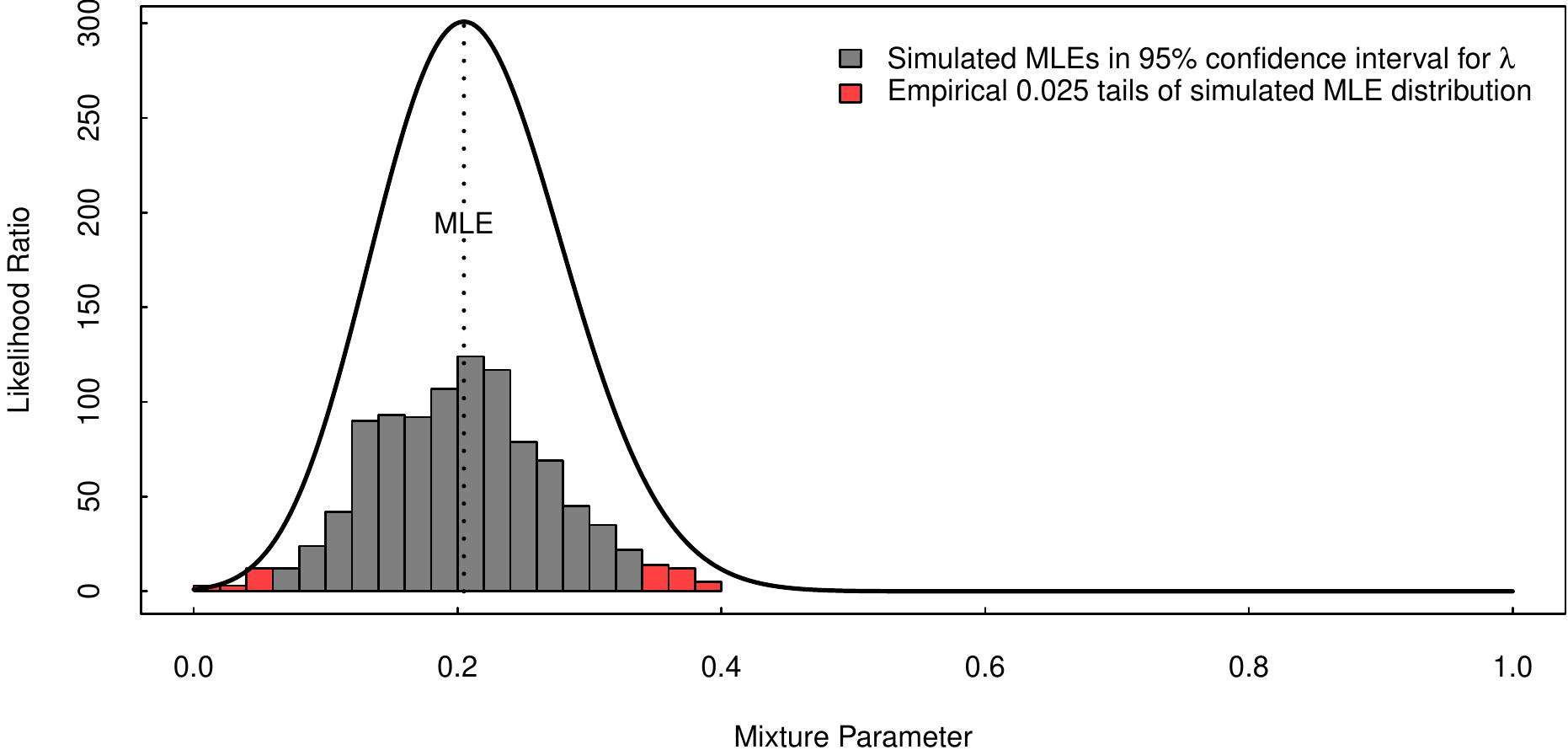}
\caption{\label{fig:likelihood}Likelihood function for $\lambda$, relative to its value at $\lambda=0$, under the assumption of a panmictic population of constant size. The vertical dashed line shows the location of the maximum likelihood estimate $\widehat{\lambda}$, and the histogram beneath the curve illustrates the empirical distribution of our simulated maximum likelihood estimates $\{\widehat{\lambda}^*\}_{k=1,\ldots,1000}$. One can see that the distribution of $\widehat{\lambda}^*$ strongly reflects the shape of the likelihood function.}
\end{figure}

\begin{figure}[h]
\centering
\includegraphics{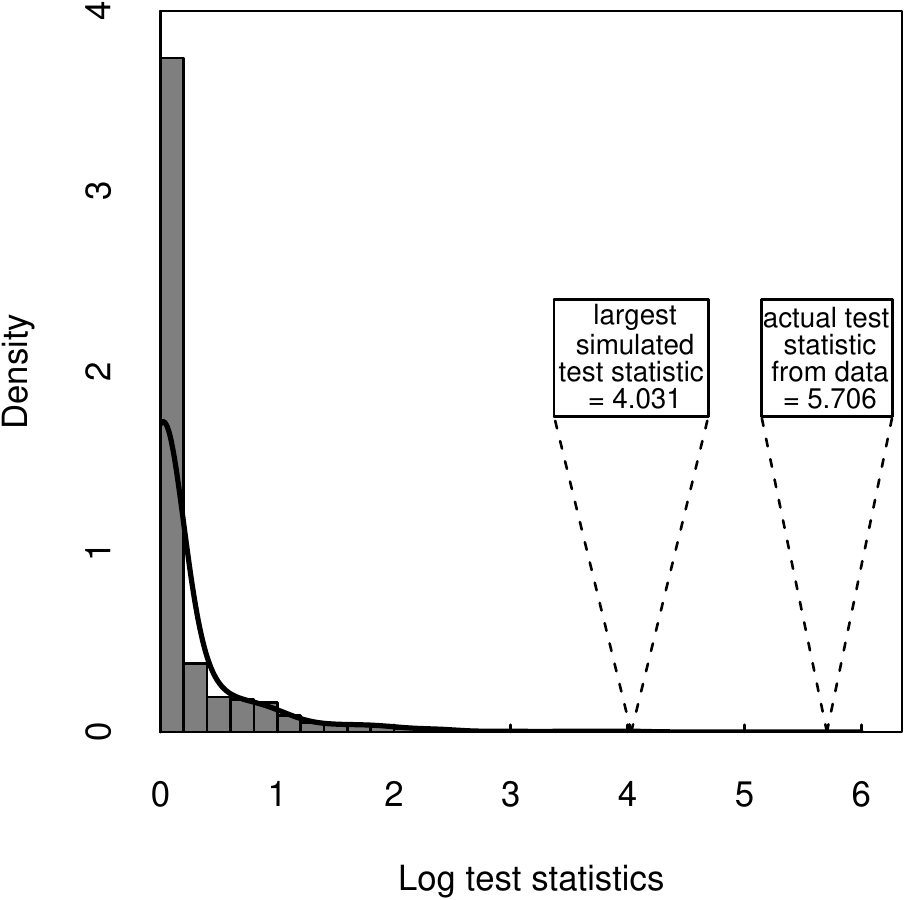}
\caption{\label{fig:Lambda}Empirical null distribution of $\ln\widetilde{\Lambda}$. A Gaussian kernel density estimate is super-imposed over the histogram of $\ln\widetilde{\Lambda}$, and the dashed lines indicate the locations of the largest of these values as well as the $\ln\widetilde{\Lambda}$ value from our data.} 
\end{figure}

\begin{figure}[h]
\centering
\includegraphics{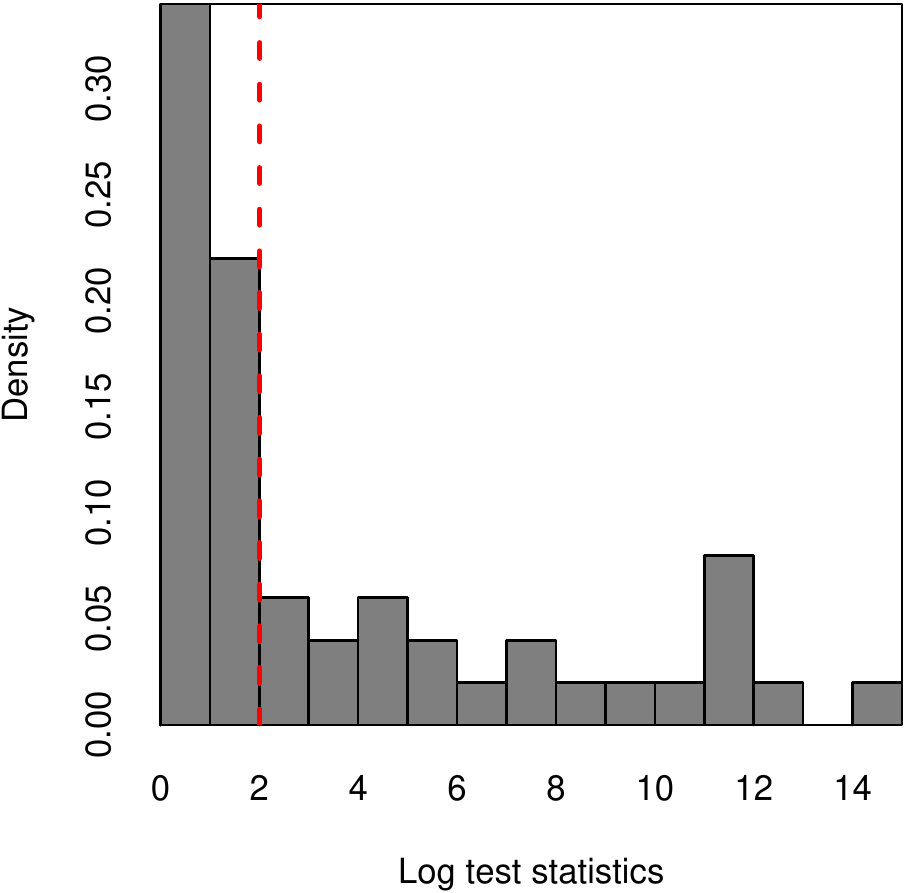}
\caption{\label{fig:Lambda-subdivision}Empirical distribution of $\ln\widetilde{\Lambda}$ under a null model incorporating complex demography. The dashed line indicated the value of the test statistic computed from the actual data, $\ln\widetilde{\Lambda} = 2.01$.}
\end{figure}

\end{document}